\documentclass[a4paper,UKenglish,cleveref, autoref, thm-restate]{lipics-v2021}

\usepackage{graphicx}
\usepackage{framed, amsmath, amsfonts, amssymb, enumitem}
\usepackage{balance,complexity,hyperref, pgfplots,tikz,tikz-3dplot,float}
\usetikzlibrary{shapes,positioning,shapes.geometric}

\usepackage[sort,nocompress]{cite}
\usepackage{subcaption}

\usepackage[noend]{algpseudocode}
\usepackage{algorithm}

\usepackage{comment}

\newtheorem{proc}[theorem]{Procedure}

\newtheorem{conj}[theorem]{Conjecture}
\newtheorem{corl}[theorem]{Corollary}

\usepackage{microtype}

\usetikzlibrary{snakes}

\def\centerarc[#1](#2)(#3:#4:#5)
{ \draw[#1] ($(#2)+({#5*cos(#3)},{#5*sin(#3)})$) arc (#3:#4:#5); }

\definecolor{pastelred}{rgb}{1.0, 0.41, 0.38}
\definecolor{applegreen}{rgb}{0.55, 0.71, 0.0}
\definecolor{amber}{rgb}{1.0, 0.75, 0.0}

\usepackage{environ}

\newcommand{\repeattheorem}[1]{%
	\begingroup
	\renewcommand{\thetheorem}{\ref{#1}}%
	\expandafter\expandafter\expandafter\theorem
	\csname reptheorem@#1\endcsname
	\endtheorem
	\endgroup
}

\NewEnviron{reptheorem}[1]{%
	\global\expandafter\xdef\csname reptheorem@#1\endcsname{%
		\unexpanded\expandafter{\BODY}%
	}%
	\expandafter\theorem\BODY\unskip\label{#1}\endtheorem
}

\newcommand{\repeatlemma}[1]{%
	\begingroup
	\renewcommand{\thelemma}{\ref{#1}}%
	\expandafter\expandafter\expandafter\lemma
	\csname replemma@#1\endcsname
	\endlemma
	\endgroup
}

\NewEnviron{replemma}[1]{%
	\global\expandafter\xdef\csname replemma@#1\endcsname{%
		\unexpanded\expandafter{\BODY}%
	}%
	\expandafter\lemma\BODY\unskip\label{#1}\endlemma
}

\newcommand{\repeatcorollary}[1]{%
	\begingroup
	\renewcommand{\thecorollary}{\ref{#1}}%
	\expandafter\expandafter\expandafter\corollary
	\csname repcorollary@#1\endcsname
	\endcorollary
	\endgroup
}

\NewEnviron{repcorollary}[1]{%
	\global\expandafter\xdef\csname repcorollary@#1\endcsname{%
		\unexpanded\expandafter{\BODY}%
	}%
	\expandafter\corollary\BODY\unskip\label{#1}\endcorollary
}

\newcommand{\repeatproposition}[1]{%
	\begingroup
	\renewcommand{\theproposition}{\ref{#1}}%
	\expandafter\expandafter\expandafter\proposition
	\csname repproposition@#1\endcsname
	\endproposition
	\endgroup
}

\NewEnviron{repproposition}[1]{%
	\global\expandafter\xdef\csname repproposition@#1\endcsname{%
		\unexpanded\expandafter{\BODY}%
	}%
	\expandafter\proposition\BODY\unskip\label{#1}\endproposition
}

\title{Recognition and Isomorphism of Proper $\boldsymbol{U}$-graphs in \emph{FPT}-time} 

\author{Deniz A\u{g}ao\u{g}lu \c{C}a\u{g}{\i}r{\i}c{\i}}{Masaryk University, Brno, Czech Republic }{agaoglu@mail.muni.cz}{https://orcid.org/0000-0002-1691-0434}{This author was supported by the Czech Science Foundation, project no.~20-04567S.}

\author{Peter Zeman}{University of Neuchâtel, Neuchâtel, Switzerland}{zeman.peter.sk@gmail.com}{https://orcid.org/0000-0003-0071-9149}{}

\authorrunning{D. A\u{g}ao\u{g}lu \c{C}a\u{g}{\i}r{\i}c{\i} and P. Zeman} 

\Copyright{Deniz A\u{g}ao\u{g}lu \c{C}a\u{g}{\i}r{\i}c{\i} and Peter Zeman} 



\category{} 

\relatedversion{} 

\acknowledgements{We would like to thank Steven Chaplick and Petr Hlin\v en\'y for their helpful ideas concerning the hardness results. We also would like to thank the anonymous referees who reviewed an earlier version and provided valuable comments and corrections that improved the quality of the manuscript.}

\nolinenumbers 

\EventEditors{John Q. Open and Joan R. Access}
\EventNoEds{2}
\EventLongTitle{42nd Conference on Very Important Topics (CVIT 2016)}
\EventShortTitle{CVIT 2016}
\EventAcronym{CVIT}
\EventYear{2016}
\EventDate{December 24--27, 2016}
\EventLocation{Little Whinging, United Kingdom}
\EventLogo{}
\SeriesVolume{42}
\ArticleNo{23}

\begin{document}
	
	\maketitle
	
	\begin{abstract}
		An \emph{$H$-graph} is an intersection graph of connected subgraphs of a suitable subdivision of a fixed graph~$H$.
		Many important classes of graphs, including interval graphs, circular-arc graphs, and chordal graphs, can be expressed as $H$-graphs, and, in particular, every graph is an $H$-graph for a suitable graph~$H$.
		An $H$-graph is called \emph{proper} if it has a representation where no  subgraph properly contains another.
		We consider the recognition and isomorphism problems for proper $U$-graphs where~$U$ is a unicylic graph.
		We prove that testing whether a graph is a (proper) $U$-graph, for some~$U$, is NP-hard.
		On the positive side, we give an \emph{FPT}-time recognition algorithm, parametrized by~$\vert U \vert$.
		As an application, we obtain an \emph{FPT}-time isomorphism algorithm for proper $U$-graphs, parametrized by~$\vert U \vert$.
		To complement this, we prove that the isomorphism problem for (proper) $H$-graphs, is as hard as the general isomorphism problem for every fixed $H$ which is not unicyclic.
	\end{abstract}
	
	\section{Introduction}
	\label{sec:intro}
	
	The concept of an $H$-graph was introduced originally by Bir\'{o}, Hujter and Tuza in 1992~\cite{biro}.
	This notion generalizes and relates to many important classes of graphs that are well-known in the literature, for instance, interval graphs, circular-arc graphs, and chordal graphs.
	
	A graph $G$ has an \emph{$H$-representation} if $G$ can be represented as an intersection graph of connected subgraphs of the graph $H$, i.e., if each vertex of $G$ can be assigned a connected subgraph of $H$ such that two subgraphs intersect if and only if the corresponding vertices are adjacent. A \emph{subdivision of $H$} is to replace each selected edge of $H$ with an induced path of any length.
	Then, $G$ is an \emph{$H$-graph} if there is a subdivision $H'$ of $H$ such that $G$ has an $H'$-representation.
	In this language, interval graphs are $K_2$-graphs, circular-arc graphs are $K_3$-graphs, and chordal graphs are the union of all $T$-graphs, where $T$ runs through all trees.
	
	The study of $H$-graphs was revived in 2017 by Chaplick et al.~\cite{zemanWG} (see also~\cite{zemanAlgo}) after partially answering the open question posed by Bir\'{o}, Hujter and Tuza: What is the complexity of testing whether a given graph is an $H$-graph if $H$ is some fixed graph?
	In particular, the recognition problem is NP-complete if $H$ contains the diamond graph as a minor.
	Further, for the case when $H$ is a tree, \emph{XP}-time algorithms were given.
	Apart from that several other problems, which are well-studied for many restricted graph classes, were also considered in~\cite{zemanWG, zemanAlgo}.
	Despite the considerable interest in this line of research~\cite{zeman2,DBLP:conf/esa/FominGR18,chaplick2020recognizing,telle,DBLP:journals/corr/abs-2107-10689,isoWalcom}, the original question of Bir\'{o}, Hujter and Tuza is still open.
	
	There are several relevant types of recognition problems for $H$-graphs.
	Firstly, both the graph $G$ and the graph $H$ can be a part of the input and the question is whether $G$ is an $H$-graph.
	This problem is already NP-complete even for chordal graphs and for the case when $H$ is a tree~\cite{KLAVIK201585}.
	The next type of problem is when the input is a graph $G$ and the graph $H$ is fixed -- this is the question asked by Bir\'{o}, Hujter, and Tuza.
	Further, we consider a third variant.
	To that end, let $\mathcal{H}$ be a class of graphs.
	By $\mathcal{H}$-graphs, we mean the class of graphs $G$ for which there exists a graph $H \in \mathcal{H}$ such that $G$ is an $H$-graph.
	For instance, if $\mathcal{T}$ is the class of all trees, then $\mathcal{T}$-graphs are exactly all chordal graphs~\cite{chordalityInters}.
	The recognition problem for $\mathcal{H}$-graphs, has a graph $G$ on the input and asks whether there is an $H \in \mathcal{H}$ such that $G$ is an $H$-graph.
	In the case of $\mathcal{T}$-graphs, the problem is solvable in polynomial time and well-known~\cite{recogChordaLinear,chordalityInters}.

	Two graphs $G$ and $H$ are called \textit{isomorphic}, denoted by $G \simeq H$, if there exists a bijection $f$ from $V(G)$ to $V(H)$ such that $\{ u,v \} \in E(G)$ if and only if $\{ f(u),f(v) \} \in E(H)$. The \textit{graph isomorphism problem} is to determine whether two graphs are isomorphic, and a problem is called \emph{GI-complete} if it can be reduced to the graph isomorphism problem in polynomial time. Considering $H$-graphs, the isomorphism problem for (proper) $S_d$-graphs is GI-complete when $d$ is on the input where $S_d$ is a star with $d$ rays. This implies to (proper) $T$-graphs and (proper) $U$-graphs  \cite{isoChordalGIComp}, and it can be solved in \emph{FPT}-time for all $H$-graphs when $H$ contains no cycle parameterized by the size of $H$ \cite{aaolu2019isomorphism,isoWalcom,SdT-graphs2021efficient,zemanWG}.  
	
	\subparagraph{Our results.}
	In this paper, we would like to make further steps towards the recognition and isomorphism problems for $H$-graphs.
	Before we summarize our contributions, we mention two special types of representations.
	In particular, we consider \emph{proper $H$-representations} in which we require that no connected subgraph of $H$ is properly contained in another one.
	Further, in \emph{Helly $H$-representations} every subset of subgraphs must satisfy the \emph{Helly property}; a family $\mathcal{F}$ of sets satisfies the Helly property if in each sub-collection of $\mathcal{F}$ whose sets pair-wise intersect, the common intersection is non-empty. We focus on the parameterized complexity of the recognition and isomorphism problems for $U$-graphs, where $U$ is a unicylic graph.
	
	\begin{itemize}
		\item In Section~\ref{sec:NP-hard}, we first show that deciding whether it is possible to construct a representation of a proper circular-arc graph, with certain cliques prescribed to be represented as Helly cliques, is NP-hard. We do this by giving a reduction from the 2-coloring problem for 3-uniform hypergraphs. We then reduce that problem to the recognition of $\mathcal{U}$-graphs and proper $\mathcal{U}$-graphs, where $\mathcal{U}$ is the class of all unicyclic graphs.
		
		\item Sections~\ref{sec:simple} and~\ref{sec:expoCliques} deal with the recognition problem of proper $U$-graphs, where $U$ is a fixed unicyclic graph.
		First, we focus on two easy cases of proper $U$-graphs, $i)$ proper $U$-graph recognition when the input graph is chordal, and $ii)$ proper Helly $U$-graph recognition. We give Procedure~\ref{proc:givenCliquesRecog} to recognize proper $U$-graphs in \emph{FPT}-time when the input graphs are chordal, and modify it to obtain an \emph{FPT}-time recognition algorithm for proper Helly $U$-graphs. 
		Then, in Section~\ref{sec:expoCliques}, we extend Procedure~\ref{proc:givenCliquesRecog} to all proper $U$-graphs.
		
		\item In Section~\ref{sec:isomorphism}, we apply the \emph{FPT} recognition algorithms mentioned above to test the isomorphism of proper $U$-graphs in \emph{FPT}-time. Finally in Section~\ref{sec:GI-completeness}, we complement this by proving that if $H$ is not unicyclic, then testing the isomorphism of $H$-graphs and proper $H$-graphs is GI-complete, i.e., as hard as the general isomorphism problem.
	\end{itemize}
	
	\subparagraph{Preliminaries.} Throughout the paper, we denote the order (i.e. the number  of vertices) of a graph $G$ by $\vert G \vert$ and the number of edges of $G$ by $\vert E(G) \vert$. We denote by \emph{$U$} the fixed unicyclic graph which means that $U$ contains exactly one cycle and several paths and trees connected to this cycle, and the order $\vert U \vert$ of $U$ is our parameter. We call the unique cycle of $U$ the \emph{``circle''} and the vertices of $U$ the \emph{``nodes''} not to be confused with the cycles and the vertices of a $U$-graph $G$, respectively.
	
	A graph $G$ of order $n$ is called \emph{complete} and denoted by $K_n$ if all vertices in $G$ are pairwise adjacent. Any complete subgraph of $G$ is called a \emph{clique}, and a clique of $G$ is called \emph{maximal} if it cannot be extended to a larger clique with the addition of new vertices.
	
	A \emph{proper $U$-graph} is a $U$-graph which has a proper representation.
	An \emph{interval graph} is the intersection graph of a set of intervals on the real line. They can be recognized and tested for isomorphism in linear time \cite{recogIntervalLinear}. Moreover, proper interval graphs can also be recognized and tested for isomorphism in linear time \cite{DBLP:journals/ipl/Keil85,DBLP:journals/dmtcs/CurtisLMNSSS13}. 
	
	A \emph{circular-arc graph} is the intersection graph of a set of arcs around a circle. On circular-arc graphs, the recognition problem can be solved in linear time \cite{recogCARCLinear} and the isomorphism problem can be solved in polynomial time \cite{DBLP:journals/corr/abs-1904-04501}.   Moreover, both the recognition and isomorphism problems are linear time solvable on proper circular-arc graphs~\cite{DBLP:journals/dmtcs/CurtisLMNSSS13,DBLP:journals/siamcomp/DengHH96}. Therefore, we assume that $U \neq K_3$ since $K_3$-graphs are circular-arc graphs.
	
	Considering a particular $U$-representation of a $U$-graph $G$, a clique $C$ of $G$ is called \emph{Helly} if it satisfies the Helly property, i.e. in that $U$-representation of $G$, all vertices of $C$ mutually intersect at some point of $U$. A \emph{Helly $U$-graph} is a $U$-graph which has some $U$-representation where all its cliques satisfy the Helly property. We note here that circular-arc graphs which are not Helly may have exponentially many maximal cliques which do not correspond to the points of the circle \cite{helly}, and we call such cliques of also proper $U$-graphs \emph{non-Helly}. In addition, the size of a non-Helly clique is at least $3$ by the definition of the Helly property. Assuming that $U \neq K_3$ is not a restriction since Helly circular-arc graphs can be recognized and tested for isomorphism in linear time \cite{DBLP:conf/cocoon/LinS06,DBLP:journals/dmtcs/CurtisLMNSSS13}.
	
	An induced cycle of length at least $4$ (a cycle in which no non-consecutive vertices are adjacent) is called a \emph{hole}. A graph is called \emph{chordal} if it contains no hole. Chordal graphs can be recognized in linear time \cite{recogChordaLinear} while the isomorphism problem on chordal graphs is \emph{GI-complete}~\cite{isoChordalGIComp} which means that it is polynomial time reducible to the general isomorphism problem which is not known to be solvable polynomial time or NP-complete. A chordal graph is also defined as the
	intersection graph of subtrees of some suitable tree $T$~\cite{chordalityInters}. When $T$ is on the input, $T$-graph recognition is NP-complete~\cite{KLAVIK201585} and $T$-graph isomorphism is \emph{GI-complete}~\cite{isoChordalGIComp}. However, for a fixed tree $T$, $T$-graph recognition can be solved in \emph{XP}-time \cite{zemanWG} and $T$-graph isomorphism can be solved in \emph{FPT}-time \cite{isoWalcom,zeman2}. Considering proper $T$-graphs, both the recognition and the isomorphism problems are solvable in \emph{FPT}-time \cite{chaplick2020recognizing,SdT-graphs2021efficient}. Therefore, we also assume that $U \neq T$ for any tree $T$.
	
	For a graph $G$, and a subgraph $C$ of $G$, the \emph{attachment} of a vertex $v \in G \setminus C$ in $C$ is the neighborhood of $v$ in $C$. Analogously, the \emph{attachments} of a connected component $X$ in $C$ is the set of neighborhoods of the vertices $v \in X$ in $C$. By the \emph{upper attachments} of $X$ in $C$, we refer to the set of maximum neighborhoods by inclusion in $C$ among the vertices of $X$.
	
	\section{NP-hardness of $\boldsymbol{U}$-graph and proper $\boldsymbol{U}$-graph recognition}
	\label{sec:NP-hard}
	
	In this section, we show that the recognition problem is NP-hard for $\mathcal{U}$-graphs and proper $\mathcal{U}$-graphs where $\mathcal{U}$ is the class of unicyclic graphs. A \emph{hypergraph} is a graph where an edge called, an \emph{hyperedge}, can join more than two vertices. If every edge of a hypergraph $H$ joins $k$ vertices, then $H$ is called \emph{$k$-uniform}. \emph{Hypergraph $c$-coloring} is the problem to assign colors to vertices so that no hyperedge joins all vertices of the same color.
	
	We define the problem {\sc HellyCliquesCARC}$(G; \mathcal{C})$ as follows: Given a circular-arc graph $G$ and a set $\mathcal{C}$ of cliques of $G$, decide whether there exists a circular-arc representation of $G$ such that all cliques in $\mathcal{C}$ are Helly cliques. We first give the following.
	
	\begin{replemma}{hardness}\label{lem:hardness}
		The problem {\sc HellyCliquesCARC}$(G; \mathcal{C})$ is NP-hard. Furthermore, this hardness result holds even if we restrict to circular-arc graphs that are the complements of perfect matchings for which every circular-arc
		representation is proper.
	\end{replemma}
	
	\begin{proof}
		The 2-coloring problem for 3-uniform hypergraphs was shown to be NP-hard \cite{DBLP:conf/focs/DinurRS02}, and we reduce it to the problem {\sc HellyCliquesCARC}$(G; \mathcal{C})$. Let $H$ be a 3-uniform hypergraph on $n$ vertices, and the circular-arc graph $G$ on $4n+4$ vertices be the complement of a perfect matching, i.e. the complete graph on $4n+4$ vertices from which a set of $2n+2$ edges forming a perfect matching is removed. For $i \in \{-1, 0, 1, \dots, 2n\}$, let $e_i=\{ v^1_{i},v^2_{i} \}$ be an edge of this removed perfect matching, and the vertices $v^1_{i}$ and $v^2_{i}$ be represented with two non-intersecting half-arcs around a circle in $G$. We construct $G$ from $H$ as follows:
		\begin{enumerate}
			\item $v^1_{\text{-}1}$,$v^2_{\text{-}1}$ is \emph{the fixing pair} used to fix the side around the circle where $v^1_{\text{-}1}$ is the top arc.
			
			\item $v^1_{i}$,$v^2_{i}$ for $i \in \{0, \dots, n\}$ are \emph{the ordering pairs} used to determine the ordered positions of the next pairs where $v^1_{0}$ is to the left of $v^1_{\text{-}1}$.
			
			\item $v^1_{i}$,$v^2_{i}$ for $i \in \{n \text{+} 1, \dots, 2n\}$ are \emph{the hypergraph pairs} used to model the vertices of $H$. 
		\end{enumerate}
		
		In $G$, every vertex $u_i$ of $H$ has to representatives $v^1_{i}$ and $v^2_{i}$, thus there is a Helly clique corresponding to every triple of vertices of $H$ since $v^1_{i}$ and $v^2_{i}$ can freely change their positions. We now fix the following triples to be the Helly cliques $\mathcal{C}$ in $G$:
		\begin{enumerate}
			\item $ \{v^1_{\text{-}1},  v^1_{i}, v^1_{i \text{+} 1} \}$, 
			$ \{v^1_{\text{-}1},  v^1_{i}, v^2_{i \text{+} 1} \}$, 
			$ \{v^1_{\text{-}1},  v^2_{i}, v^2_{i \text{+} 1} \}$ for $i \in \{0, \dots, n \text{-} 1\}$.
			
			\item $ \{v^2_{i \text{-} 1},  v^1_{i}, v^1_{n \text{+} i} \}$, 
			$ \{v^2_{i \text{-} i},  v^2_{i}, v^1_{n \text{+} i} \}$ for $i \in \{1, \dots, n\}$.
			
			\item $ \{v^1_{n \text{+} i},  v^1_{n \text{+} j}, v^2_{n \text{+} k} \}$ for every hyperedge $ \{u_i, u_j, u_k\}$.
		\end{enumerate}
		
		After we fix the Helly cliques, $v^1_{i}$ and $v^2_{i}$ can change the positions for no $i$ since the representation is fixed. Now, on the arc $ v^1_{ \text{-} 1}$, we can not have \emph{LLR} and \emph{RRL} configurations for the arcs $ \{v^1_{n \text{+} i},  v^1_{n \text{+} j}, v^2_{n \text{+} k} \}$ where $L$ and $R$ denote the left and right, respectively. Therefore, for the arcs $ \{v^1_{n \text{+} i},  v^1_{n \text{+} j}, v^1_{n \text{+} k} \}$, the forbidden configurations are \emph{LLL} and \emph{RRR} which means that there is no monochromatic edge joining the vertices $ \{u_i, u_j, u_k \}$ of $H$. As a result, the problem whether the fixed triples are Helly cliques in a circular-arc representation of $G$ is NP-hard since the 2-coloring problem for 3-uniform hypergraphs is NP-hard \cite{DBLP:conf/focs/DinurRS02}. Moreover, the constructed circular-arc graph is proper since no half-arc contains another. Therefore, the problem {\sc HellyCliquesCARC}$(G; \mathcal{C})$ is NP-hard also for proper circular-arc graphs.
	\end{proof}

	\begin{figure}[h]
	\centering
	
	\begin{subfigure}[t]{0.35\linewidth}
		\centering
		\begin{tikzpicture}[xscale=0.8,yscale=0.7]
			
			\draw[dashed, ultra thick] (6,1) circle (0.9cm);
			
			\draw[ultra thick,lightgray] (8.05,0) arc (-25:145:2.3cm);
			\node [label=left:{\textcolor{lightgray}{\tiny{ $v^1_{2n}$ }}}] (11) at (4.62,2.22) { };
			\draw[ultra thick,gray] (3.9,2) arc (155:325:2.3cm);
			\node [label=left:{\textcolor{gray}{\tiny{ $v^2_{2n}$ }}}] (12) at (8.6,-0.1) { };
			
			\draw[ultra thick,cyan] (7.07,1.2) arc (10:170:1.1cm);
			\node [label=left:{\textcolor{cyan}{\tiny{ $v^1_{\text{-}1}$ }}}] (1) at (5.5,1) { };	
			\draw[ultra thick,blue] (4.9,0.8) arc (190:350:1.1cm);
			\node [label=left:{\textcolor{blue}{\tiny{ $v^2_{\text{-}1}$ }}}] (2) at (7.75,1) { };
			
			\draw[ultra thick,green] (7.08,1.85) arc (40:200:1.4cm);
			\node [label=left:{\textcolor{green}{\tiny{ $v^1_{0}$ }}}] (1) at (5.3,0.32) { };
			\draw[ultra thick,teal] (4.9,0.15) arc (220:380:1.4cm);
			\node [label=left:{\textcolor{teal}{\tiny{ $v^2_{0}$ }}}] (1) at (7.7,1.7) { };
			
			\draw[ultra thick,yellow] (6.57,2.51) arc (70:235:1.65cm);
			\node [label=left:{\textcolor{yellow}{\tiny{ $v^1_{1}$ }}}] (1) at (5.8,-0.43) { };
			\draw[ultra thick,orange] (5.39,-0.5) arc (250:415:1.65cm);
			\node [label=left:{\textcolor{orange}{\tiny{ $v^2_{1}$ }}}] (1) at (7.3,2.49) { };
			
			\draw[ultra thick,pink] (5.75,2.86) arc (95:265:1.87cm);
			\node [label=left:{\textcolor{pink}{\tiny{ $v^1_{n+i}$ }}}] (1) at (6.9,-0.9) { };
			\draw[ultra thick,magenta] (6.2,-0.85) arc (275:445:1.87cm);
			\node [label=left:{\textcolor{magenta}{\tiny{ $v^2_{n+i}$ }}}] (1) at (6.865,2.8) { };
			
			\draw[dotted] (5.4,2.5) -- (5.25,2.75)  node[midway, above]{};
			\draw[dotted] (5.1,2.35) -- (4.95,2.6)  node[midway, above]{};
			\draw[dotted] (4.8,2.15) -- (4.65,2.4)  node[midway, above]{};
			\draw[dotted] (4.6,1.9) -- (4.45,2.15)  node[midway, above]{};
			\draw[dotted] (4.4,1.6) -- (4.25,1.85)  node[midway, above]{};
			\draw[dotted] (4.3,1.25) -- (4.15,1.5)  node[midway, above]{};
			
			\draw[dotted] (5.1,2.75) -- (4.95,3)  node[midway, above]{};
			\draw[dotted] (4.8,2.6) -- (4.65,2.85)  node[midway, above]{};
			\draw[dotted] (4.55,2.4) -- (4.4,2.65)  node[midway, above]{};
			\draw[dotted] (4.35,2.2) -- (4.2,2.45)  node[midway, above]{};
			\draw[dotted] (4.1,1.65) -- (3.95,1.9)  node[midway, above]{};
			\draw[dotted] (4,1.35) -- (3.85,1.6)  node[midway, above]{};
			
			\node (11) at (6,-2.1) {\textbf{(a)}};
			\label{a}
		\end{tikzpicture}
		
		\label{fig:v0,1,(n+i)}
	\end{subfigure}
	\quad
	\begin{subfigure}[t]{0.25\linewidth}
		\centering
		\begin{tikzpicture}[xscale=0.8,yscale=0.7]
			
			\draw[dashed, ultra thick] (6,1) circle (0.9cm);
			
			\draw[ultra thick,cyan] (7.07,1.2) arc (10:170:1.1cm);
			\node [label=left:{\textcolor{cyan}{\tiny{ $v^1_{\text{-}1}$ }}}] (1) at (5.5,1) { };	
			\draw[ultra thick,blue] (4.9,0.8) arc (190:350:1.1cm);
			\node [label=left:{\textcolor{blue}{\tiny{ $v^2_{\text{-}1}$ }}}] (2) at (7.75,1) { };
			
			\draw[ultra thick,green] (7.08,1.85) arc (40:200:1.4cm);
			\node [label=left:{\textcolor{green}{\tiny{ $v^1_{0}$ }}}] (3) at (5.3,0.32) { };
			\draw[ultra thick,teal] (4.9,0.15) arc (220:380:1.4cm);
			\node [label=left:{\textcolor{teal}{\tiny{ $v^2_{0}$ }}}] (4) at (7.7,1.7) { };
			
			\draw[ultra thick,yellow] (6.57,2.51) arc (70:235:1.65cm);
			\node [label=left:{\textcolor{yellow}{\tiny{ $v^1_{1}$ }}}] (5) at (5.8,-0.43) { };	
			\draw[ultra thick,orange] (5.39,-0.5) arc (250:415:1.65cm);
			\node [label=left:{\textcolor{orange}{\tiny{ $v^2_{1}$ }}}] (6) at (7.3,2.51) { };
			
			\draw[dotted] (5.4,2.5) -- (5.25,2.75)  node[midway, above]{};
			\draw[dotted] (5.1,2.35) -- (4.95,2.6)  node[midway, above]{};
			\draw[dotted] (4.8,2.15) -- (4.65,2.4)  node[midway, above]{};
			\draw[dotted] (4.6,1.9) -- (4.45,2.15)  node[midway, above]{};
			\draw[dotted] (4.4,1.6) -- (4.25,1.85)  node[midway, above]{};
			\draw[dotted] (4.3,1.25) -- (4.15,1.5)  node[midway, above]{};
			
			\filldraw[draw=gray, fill=gray, opacity=0.2] (5.95,3) rectangle (6.05,1.5);
			\node [label=left:{\tiny{ \{ $v^1_{\text{-}1}$,$v^1_{0}$, $v^1_{1}$ \} }}] (7) at (7.2,3.3) { };
			
			\filldraw[draw=gray, fill=gray, opacity=0.2, rotate around={132:(6.54,2.27)}] (5.95,3) rectangle (6.05,1.5);
			\node [label=left:{\tiny{ \{ $v^1_{\text{-}1}$,$v^1_{0}$, $v^2_{1}$ \} }}] (7) at (9.54,2.1) { };
			
			\filldraw[draw=gray, fill=gray, opacity=0.2, rotate around={105:(6.94,2.27)}] (5.95,3) rectangle (6.05,1.5);
			\node [label=left:{\tiny{ \{ $v^1_{\text{-}1}$,$v^2_{0}$, $v^2_{1}$ \} }}] (7) at (9.83,1.25) { };
			
			\node (11) at (6,-2.1) {\textbf{(b)}};
			\label{b}
		\end{tikzpicture}
		
		\label{fig:v-1,0,1}
	\end{subfigure}
	\quad
	\begin{subfigure}[t]{0.25\linewidth}
		\centering
		\begin{tikzpicture}[xscale=0.8,yscale=0.7]
			
			\draw[dashed, ultra thick] (6,1) circle (0.9cm);
			
			\draw[ultra thick,green] (7.08,1.85) arc (40:200:1.4cm);
			\node [label=left:{\textcolor{green}{\tiny{ $v^1_{0}$ }}}] (1) at (5.3,0.32) { };
			\draw[ultra thick,teal] (4.9,0.15) arc (220:380:1.4cm);
			\node [label=left:{\textcolor{teal}{\tiny{ $v^2_{0}$ }}}] (2) at (7.7,1.7) { };
			
			\draw[ultra thick,yellow] (6.57,2.51) arc (70:235:1.65cm);
			\node [label=left:{\textcolor{yellow}{\tiny{ $v^1_{1}$ }}}] (3) at (5.8,-0.43) { };
			\draw[ultra thick,orange] (5.39,-0.5) arc (250:415:1.65cm);
			\node [label=left:{\textcolor{orange}{\tiny{ $v^2_{1}$ }}}] (4) at (7.3,2.49) { };
			
			\draw[ultra thick,pink] (5.75,2.86) arc (95:265:1.87cm);
			\node [label=left:{\textcolor{pink}{\tiny{ $v^1_{n+1}$ }}}] (5) at (6.9,-0.9) { };
			\draw[ultra thick,magenta] (6.2,-0.85) arc (275:445:1.87cm);
			\node [label=left:{\textcolor{magenta}{\tiny{ $v^2_{n+1}$ }}}] (6) at (6.865,2.8) { };
			
			\draw[dotted] (5.4,2.5) -- (5.25,2.75)  node[midway, above]{};
			\draw[dotted] (5.1,2.35) -- (4.95,2.6)  node[midway, above]{};
			\draw[dotted] (4.8,2.15) -- (4.65,2.4)  node[midway, above]{};
			\draw[dotted] (4.6,1.9) -- (4.45,2.15)  node[midway, above]{};
			\draw[dotted] (4.4,1.6) -- (4.25,1.85)  node[midway, above]{};
			\draw[dotted] (4.3,1.25) -- (4.15,1.5)  node[midway, above]{};
			
			\draw[dotted] (5.1,2.75) -- (4.95,3)  node[midway, above]{};
			\draw[dotted] (4.8,2.6) -- (4.65,2.85)  node[midway, above]{};
			\draw[dotted] (4.55,2.4) -- (4.4,2.65)  node[midway, above]{};
			\draw[dotted] (4.35,2.2) -- (4.2,2.45)  node[midway, above]{};
			\draw[dotted] (4.1,1.65) -- (3.95,1.9)  node[midway, above]{};
			\draw[dotted] (4,1.35) -- (3.85,1.6)  node[midway, above]{};
			
			\filldraw[draw=gray, fill=gray, opacity=0.2, rotate around={135:(5.95,0.84)}] (5.95,3) rectangle (6.05,1.3);
			\node [label=left:{\tiny{ \{ $v^2_{0}$,$v^1_{1}$, $v^1_{n+1}$ \} }}] (7) at (5.6,-0.95) { };
			
			\filldraw[draw=gray, fill=gray, opacity=0.2, rotate around={168:(5.95,0.86)}] (5.95,3) rectangle (6.05,1.3);
			\node [label=left:{\tiny{ \{ $v^2_{0}$,$v^2_{1}$, $v^1_{n+1}$ \} }}] (8) at (6.7,-1.5) { };

			\node (11) at (6,-2.1) {\textbf{(c)}};
			\label{c}
		\end{tikzpicture}
		
		\label{fig:v0,1,(n+i)}
	\end{subfigure}
	
	\vspace*{-2.5ex}
	\caption{An illustration of the construction given in the proof of Lemma~\ref{lem:hardness}. (a) A circular-arc graph $G$ which is the complement of the perfect matching $v^1_{i}$,$v^2_{i}$ for $i \in \{\text{-}1, 0, \dots, 2n\}$ where $v^1_{i}$ and $v^2_{i}$ are represented with two non-intersecting half-arcs around the circle and the dotted lines correspond to the intermediate half-arcs. (b) The fixed (non-) Helly triples $ \{v^1_{\text{-}1},  v^1_{0}, v^1_{1} \}$, $ \{v^1_{\text{-}1},  v^1_{0}, v^2_{1} \}$	and $ \{v^1_{\text{-}1}, v^2_{0}, v^2_{1} \}$. (c) The fixed Helly triples $ \{v^2_{0},  v^1_{1}, v^1_{n \text{+} 1} \}$	and $ \{v^2_{0}, v^2_{1}, v^1_{n \text{+} 1} \}$.}
	\label{fig:reduction} 
\end{figure}
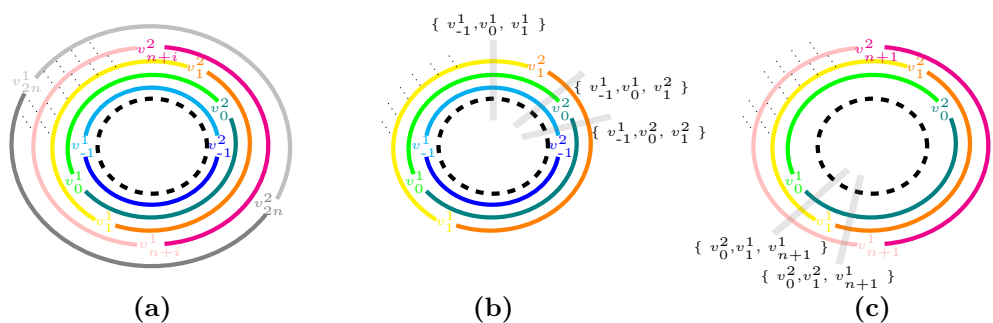
	
	By Lemma~\ref{lem:hardness}, we get the following.
	\begin{corl}\label{corl:properUrecogNP-hard}
		The recognition problem is NP-complete for $\mathcal{U}$-graphs and proper $\mathcal{U}$-graphs where $\mathcal{U}$ is the class of all unicyclic graphs.
	\end{corl}
	
\begin{proof}
	Let $G$ be a circular-arc graph and $\mathcal{S}$ be a set of its cliques. We construct a graph $G^*$ in polynomial time as follows: for each clique $C \in \mathcal{S}$, add a new vertex $v$ adjacent to all vertices in $C$. Then, for each $C \in \mathcal{S}$, $C \cup v$ is a maximal clique in $G^*$. Let $\mathcal{S'}$ denote the set of those new maximal cliques of $G^*$. Now, $G^*$ is a $U$-graph for a unicyclic graph $U \in \mathcal{U}$ which has $\vert \mathcal{S} \vert$ branching nodes with $\mathcal{S'}$ placed on them if all cliques in $\mathcal{S}$ can be represented as Helly cliques in a circular-arc representation of $G$. Since the latter is NP-hard by Lemma~\ref{lem:hardness}, the recognition of $\mathcal{U}$-graphs is NP-hard. Moreover, if $G$ is a proper circular-arc graph, then the constructed graph $G^*$ is a proper $U$-graph if all cliques in $\mathcal{S}$ can be represented as Helly cliques in a circular-arc representation of $G$ when the new vertex $v$ added for each $C \in \mathcal{S}$ is represented on the edge not from the circle of $U$ incident to the branching node which $C \cup v$ is placed on. Therefore, the recognition of proper $\mathcal{U}$-graphs is also NP-hard.
\end{proof}

	\section{Recognizing proper Helly and proper chordal $\boldsymbol{U}$-graphs}
	\label{sec:simple}
	
	In this section, we give two \emph{FPT}-procedures to decide whether $i)$ a given chordal graph $G$ is a proper $U$-graph, and $ii)$ a given graph $G$ is a proper Helly $U$-graph. In a proper $U$-representation of $G$ (if exists), every branching clique can be extended to a maximal clique~\cite{zemanWG} and we refer to the maximal cliques of $G$ placed on the branching nodes (i.e. the nodes of degree at least $3$) of $U$ as the \emph{branching cliques}.
	
	We first consider the case when the input graph $G$ is chordal. It is known that chordal graphs can be recognized in linear time and they have linearly many maximal cliques which can be listed in linear time \cite{recogChordaLinear}. Therefore, given a graph which is chordal, we can list its set of maximal cliques efficiently, and among these maximal cliques, our aim is to identify a bounded number of maximal cliques which can exclusively be placed on the branching nodes of the circle in \emph{FPT}-time parameterized by $\vert U \vert$. The following was proven in \cite{SdT-graphs2021efficient}. 
	
	\begin{lemma}{\bf\cite[Lemma~7.3]{SdT-graphs2021efficient}}\label{lem:branchCliqueForTCliques}
		Given a proper $T$-graph $G$ for a fixed tree $T$, one can identify an isomorphism-invariant set of bounded number of maximal cliques which can be used as branching cliques placed on branching nodes of $T$ in some proper $T$-representations of $G$ in \emph{FPT}-time parameterized by $\vert T \vert$.
	\end{lemma}

	The following lemma shows that also for proper $U$-graphs for a fixed unicyclic graph $U$, we can identify a bounded number of maximal cliques which can be used as branching cliques among polynomially many maximal cliques.
	
	\begin{replemma}{richBranchBczTpartsCliques}\label{lem:richBranchBczTpartsCliques}
		 For a fixed unicyclic graph $U$, given a proper $U$-graph $G$ which has polynomially many maximal cliques, one can identify an isomorphism-invariant bounded number of maximal cliques which can be used as branching cliques placed on the circle of $U$ in some proper $U$-representations of $G$ in \emph{FPT}-time parameterized by $\vert U \vert$.
	\end{replemma}

	\begin{proof}
		We can perceive a proper $U$-representation of $G$ as the assignment of its maximal cliques to the nodes and edges of $U$ where each node is a single maximal clique and each edge can have many maximal cliques on it. For every proper $U$-representation of $G$ and any edge $e=\{ v,w \}$ of the circle of $U$, let $G'$ be the subgraph of $G$ obtained by discarding $\mathcal{C}_e \setminus (C_v \cup C_w)$ where $\mathcal{C}_e$ denotes the set of maximal cliques placed on $e$, and $C_v$ and $C_w$ denote the maximal cliques placed on the branching nodes $v$ and $w$ of the circle of $U$, respectively. Note that every such $G'$ is chordal and a proper $T$-graph where $T=U-e$ is a tree since $U$ is unicyclic. Moreover, every $G'$ is actually an induced subgraph of $G$ except the case $C_v \cup C_w \supseteq C_x$ for all $C_e \in \mathcal{C}_x \setminus (C_v \cup C_w)$. Then, for each such $G'$, one can identify a bounded number of maximal cliques of $G'$ which can be used as branching cliques in some proper $T$-representation of $G'$ in \emph{FPT}-time parameterized by $\vert T \vert$ by Lemma~\ref{lem:branchCliqueForTCliques}. 
		Therefore, one can identify a bounded number of maximal cliques of $G$ which can be used as branching cliques in some proper $U$-representation of $G$ in \emph{FPT}-time parameterized by $\vert U \vert$ using the algorithm of \cite{chaplick2020recognizing}. The arguments on our selection being isomorphism-invariant and capturing some proper $T$-representations from Lemma~\ref{lem:branchCliqueForTCliques} also hold here. 
	\end{proof}

	We call the maximal cliques identified using Lemma~\ref{lem:richBranchBczTpartsCliques} the \emph{rich cliques}, and since their number in proper $U$-graphs is bounded by $\vert U \vert$ by Lemma~\ref{lem:richBranchBczTpartsCliques}, we can afford in \emph{FPT}-time to try all possible assignments of rich cliques to the branching nodes on the circle of $U$ and obtain an \emph{FPT}-time recognition algorithm. Given a graph $G$, we make the following assumptions:
	
	\begin{itemize}
		\item [\textcolor{gray}{\textbf{A1.}}] $G$ is not a proper $U_1$-graph for some subgraph or minor $U_1$ of $U$. This is due to the fact that if $U$ does not contain a vertex of degree $2$, which is not a restriction since a $U$-representation can be obtained by applying any suitable subdivision on $U$, then it is not a subdivision of some other unicyclic graph $U^*$, and the number of all subgraphs and minors of $U$ is bounded by a function of $\vert U \vert$. Therefore, Procedure~\ref{proc:givenCliquesRecog} can be applied to test whether $G$ is a proper $U_1$-graph for each subgraph and minor $U_1$ of $U$.
		
		\item [\textcolor{gray}{\textbf{A2.}}] $G$ is connected since every component of a disconnected proper $U$-graph without a rich clique is an interval graph and can be placed on any edge incident to a leaf of $U$. Otherwise, the number of rich cliques, therefore the number of such components, is bounded and, we can try all possible assignments of such components to the subgraphs of $U$.
	\end{itemize}
	
	\begin{proc}\label{proc:givenCliquesRecog} \rm
		Given a connected chordal graph $G$ on $n$ vertices and a fixed unicyclic graph $U$, we decide whether $G$ is a proper $U$-graph as follows:
		
		\begin{enumerate}
			\item Let $\mathcal{C}$ be the set of all maximal cliques of $G$, $\mathcal{B}$ denote the branching nodes on the circle of $U$ and $deg(\mathcal{B})$ be the sum of degrees of nodes of $U$ in $\mathcal{B}$. Find the set $\mathcal{C^*} \subseteq \mathcal{C}$ of its rich cliques using Lemma~\ref{lem:branchCliqueForTCliques} and Lemma~\ref{lem:richBranchBczTpartsCliques}.
			If $\vert \mathcal{C^*} \vert < \vert \mathcal{B} \vert$ or $\vert \mathcal{C^*} \vert$ is not bounded by a function of $\vert U \vert$ by Lemma~\ref{lem:richBranchBczTpartsCliques}, return that $G$ is not a proper $U$-graph.
			
			\item For each assignment $f: \mathcal{C^*} \rightarrow \mathcal{B}$, find the connected components $\mathcal{X}$ of $G-f(\mathcal{C^*})$. If there are more than $deg(\mathcal{B}) - \vert \mathcal{B} \vert$ connected components in $\mathcal{X}$, then move on to another assignment. Otherwise, $\vert \mathcal{X} \vert \leq deg(\mathcal{B}) - \vert \mathcal{B} \vert$ and we proceed as follows:
			
			\begin{enumerate}
				
				\item For each branching node $b_i \in \mathcal{B}$:
				\begin{enumerate}
					
					\item Let $C_i$ denote the maximal clique of $\mathcal{C}^*$ placed on the branching node $b_i \in \mathcal{B}$ in the current assignment $f$, i.e. $C_i = f{\mid_{b_i}}$. 
					
					\item Let $P$ denote the maximal induced subgraph of the circle containing $b_i$ and all other branching nodes $b_j$ such that $C_i \cap C_j \neq \emptyset$, i.e. $C_i$ and $C_j$ share a vertex.
					
					\item \textbf{If $\boldsymbol{P}$ does not contain at least one branching node $\boldsymbol{b_j}$}, then it is not a circle. Let $Y \subsetneq U$ denote the maximal connected subtree of $U$ which contain only $P$ from the circle, and $\mathcal{Y^*}$ denote the union of the branching cliques and the connected components placed on $Y$, and the additional components $X_k$ and $X_l$ placed on the edges next to the ends of $P$ on the circle (we prove that such components must exist). If $G[\mathcal{Y^*}]$ is not a proper $Y$-graph, then move on to another assignment.
					
					\item \textbf{Else if $\boldsymbol{P}$ is a circle and there is no connected components placed on the edges of $\boldsymbol{P}$}, then
					if there exists no pair of consecutive branching nodes on $P$ such that the rich cliques placed on them share the minimal (at most $2$) sets of common vertices with $C_i$, then move on to another assignment. Otherwise, let $b_{j}$ and $b_{j+1}$ be the furthest branching nodes from $b_i$ such that $C_j \cap C_i \subsetneq C_k \cap C_i$ for all $b_k$ between $b_i$ and $b_j$, and $C_{j+1} \cap C_i \subsetneq C_k \cap C_i$ for all $b_k$ between $b_i$ and $b_{j+1}$. Let $e$ denote the edge between $b_j$ and $b_{j+1}$, $Y = U - e$ denote the maximally connected subtree of $U$ which contains $P - e$ from the circle, and $\mathcal{Y^*} = G - (C_j \cap C_{j+1})$. If $G[\mathcal{Y^*}]$ is a proper $Y$-graph, return that $G$ is a proper $U$-graph. Otherwise, move on to another assignment.
					
					\item \textbf{Else, $\boldsymbol{P}$ is a circle and there is exactly one connected component $\boldsymbol{X_k} \in \boldsymbol{\mathcal{X}}$ placed on the edge $\boldsymbol{e}$ of $\boldsymbol{P}$} between the branching nodes $b_{j}$ and $b_{j+1}$. Let $Y = U - e$ denote the maximal connected subtree of $U$ which contain only $P - e$ from the circle, and $\mathcal{Y^*}$ denote the union of the branching cliques and the connected components placed on $Y$, and additionally two copies of $X_k$ one with its attachment in $C_j$ and the other with its attachment in $C_{j+1}$. If $G[\mathcal{Y^*}]$ is a proper $Y$-graph and $G[C_j \cup X_k \cup C_{j+1}]$ is a proper interval graph, return that $G$ is a proper $U$-graph. Otherwise, move on to another assignment.
					
				\end{enumerate}
				
				\item Return that $G$ is a proper $U$-graph.	
			\end{enumerate} 
			\item Return that $G$ is not a proper $U$-graph.
		\end{enumerate}	
	\end{proc}
	
	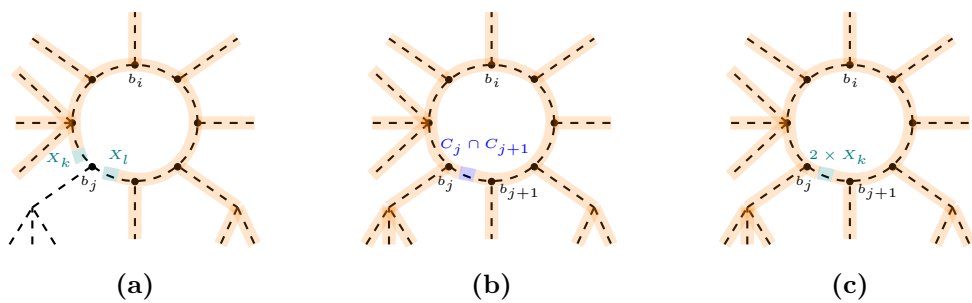
\begin{figure}[t]
		\centering
		\begin{subfigure}[t]{0.32\linewidth}
			\centering
			\begin{tikzpicture}[xscale=0.75,yscale=0.7]
				
				\draw[dashed, thick] (6,1) circle (1.1cm);
				
				\draw[dashed, thick] (6,2.1) -- (6,3.1);
				\draw[dashed, thick] (6.8,1.85) -- (7.8,2.6);
				\draw[dashed, thick] (7.1,1) -- (8.1,1);
				\draw[dashed, thick] (6.8,0.15) -- (7.8,-0.6);
				\draw[dashed, thick] (6,-0.2) -- (6,-1.2);
				\draw[dashed, thick] (4.2,-0.6) -- (5.2,0.15);
				\draw[dashed, thick] (4.9,1) -- (3.9,1);
				\draw[dashed, thick] (4.2,2.6) -- (5.2,1.85);
				
				\draw[dashed, thick] (7.8,-0.6) -- (8.1,-1.3);
				\draw[dashed, thick] (7.8,-0.6) -- (7.5,-1.3);
				
				\draw[dashed, thick] (4.6,-1.3) -- (4.2,-0.6);
				\draw[dashed, thick] (4.2,-1.3) -- (4.2,-0.6);
				\draw[dashed, thick] (3.8,-1.3) -- (4.2,-0.6);
				
				\draw[dashed, thick] (4.9,1) -- (3.9,0);
				\draw[dashed, thick] (4.9,1) -- (3.9,2);
				
				\node[label=left:{\tiny{$b_i$}}] (1) at (6.5,1.8) {};
				\node[label=left:{\tiny{$b_j$}}] (6) at (5.7,-0.15) {};
				
				\node[label=left:{\textcolor{teal}{\tiny{$X_k$}}}] (6) at (5.23,0.3) {};
				\node[label=left:{\textcolor{teal}{\tiny{$X_l$}}}] (6) at (6.25,0.36) {};
				
				\node at (6,2.1) [circle,draw, fill=black, opacity=1, color=black, inner sep=0.3mm] (11) {};
				\node at (6.75,1.83) [circle,draw, fill=black, opacity=1, color=black, inner sep=0.3mm] (12) {};
				\node at (7.1,1) [circle,draw, fill=black, opacity=1, color=black, inner sep=0.3mm] (13) {};
				\node at (6.75,0.17) [circle,draw, fill=black, opacity=1, color=black, inner sep=0.3mm] (14) {};
				\node at (6,-0.11) [circle,draw, fill=black, opacity=1, color=black, inner sep=0.3mm] (15) {};
				\node at (5.25,0.17) [circle,draw, fill=black, opacity=1, color=black, inner sep=0.3mm] (16) {};
				\node at (4.9,1) [circle,draw, fill=black, opacity=1, color=black, inner sep=0.3mm] (17) {};
				\node at (5.25,1.82) [circle,draw, fill=black, opacity=1, color=black, inner sep=0.3mm] (18) {};

				\draw[-, line width=5pt, teal, opacity=0.2] (5.44,0.06) .. controls(5.5,0.045) and (5.65,0.02) .. (5.7,0);
				
				\draw[-, line width=5pt, teal, opacity=0.2] (5,0.5) .. controls(5.03,0.4) and (5.05,0.3) .. (5.07,0.27);
				
				\draw[-, line width=5pt, orange, opacity=0.2] (5,0.5) .. controls(4.6,2.1) and (6,2.4) .. (6.75,1.85) .. controls(7.5,1) and (6.8,-0.4) .. (5.7,0);	
				
				\draw[line width=5pt, orange, opacity=0.2] (6,2.1) -- (6,3.1);
				\draw[line width=5pt, orange, opacity=0.2] (6.8,1.85) -- (7.8,2.6);
				\draw[line width=5pt, orange, opacity=0.2] (7.1,1) -- (8.1,1);
				\draw[line width=5pt, orange, opacity=0.2] (6.8,0.15) -- (7.8,-0.6);
				\draw[line width=5pt, orange, opacity=0.2] (6,-0.2) -- (6,-1.2);
				\draw[line width=5pt, orange, opacity=0.2] (4.9,1) -- (3.9,1);
				\draw[line width=5pt, orange, opacity=0.2] (4.2,2.6) -- (5.2,1.85);
				
				\draw[line width=5pt, orange, opacity=0.2] (7.8,-0.6) -- (8.1,-1.3);
				\draw[line width=5pt, orange, opacity=0.2] (7.8,-0.6) -- (7.5,-1.3);
				
				\draw[line width=5pt, orange, opacity=0.2] (4.9,1) -- (3.9,0);
				\draw[line width=5pt, orange, opacity=0.2] (4.9,1) -- (3.9,2);
				
				\node (111) at (6,-2.1) {\textbf{(a)}};
				\label{a}
			\end{tikzpicture}
			
			\label{fig:case1}
		\end{subfigure}
		~
		\begin{subfigure}[t]{0.32\linewidth}
			\centering
			\begin{tikzpicture}[xscale=0.75,yscale=0.7]
				
				\draw[dashed, thick] (6,1) circle (1.1cm);
				
				\draw[dashed, thick] (6,2.1) -- (6,3.1);
				\draw[dashed, thick] (6.8,1.85) -- (7.8,2.6);
				\draw[dashed, thick] (7.1,1) -- (8.1,1);
				\draw[dashed, thick] (6.8,0.15) -- (7.8,-0.6);
				\draw[dashed, thick] (6,-0.2) -- (6,-1.2);
				\draw[dashed, thick] (4.2,-0.6) -- (5.2,0.15);
				\draw[dashed, thick] (4.9,1) -- (3.9,1);
				\draw[dashed, thick] (4.2,2.6) -- (5.2,1.85);
				
				\draw[dashed, thick] (7.8,-0.6) -- (8.1,-1.3);
				\draw[dashed, thick] (7.8,-0.6) -- (7.5,-1.3);
				
				\draw[dashed, thick] (4.6,-1.3) -- (4.2,-0.6);
				\draw[dashed, thick] (4.2,-1.3) -- (4.2,-0.6);
				\draw[dashed, thick] (3.8,-1.3) -- (4.2,-0.6);
				
				\draw[dashed, thick] (4.9,1) -- (3.9,0);
				\draw[dashed, thick] (4.9,1) -- (3.9,2);
				
				\node[label=left:{\tiny{$b_i$}}] (1) at (6.5,1.8) {};
				\node[label=left:{\tiny{$b_j$}}] (6) at (5.7,-0.15) {};
				\node[label=right:{\tiny{$b_{j+1}$}}] (5) at (5.8,-0.3) {};
				\node[label=right:{\textcolor{blue}{\tiny{$C_j \cap C_{j+1}$}}}] (55) at (4.76,0.54) {};
				
				\node at (6,2.1) [circle,draw, fill=black, opacity=1, color=black, inner sep=0.3mm] (11) {};
				\node at (6.75,1.83) [circle,draw, fill=black, opacity=1, color=black, inner sep=0.3mm] (12) {};
				\node at (7.1,1) [circle,draw, fill=black, opacity=1, color=black, inner sep=0.3mm] (13) {};
				\node at (6.75,0.17) [circle,draw, fill=black, opacity=1, color=black, inner sep=0.3mm] (14) {};
				\node at (6,-0.11) [circle,draw, fill=black, opacity=1, color=black, inner sep=0.3mm] (15) {};
				\node at (5.25,0.17) [circle,draw, fill=black, opacity=1, color=black, inner sep=0.3mm] (16) {};
				\node at (4.9,1) [circle,draw, fill=black, opacity=1, color=black, inner sep=0.3mm] (17) {};
				\node at (5.25,1.82) [circle,draw, fill=black, opacity=1, color=black, inner sep=0.3mm] (18) {};
				
				\draw[-, line width=5pt, blue, opacity=0.2] (5.44,0.06) .. controls(5.5,0.045) and (5.65,0.02) .. (5.7,0);
				
				\draw[-, line width=5pt, orange, opacity=0.2] (5.45,0.05) .. controls(5.13,0.2) and (5.05,0.4) .. (5,0.5) .. controls(4.6,2.1) and (6,2.4) .. (6.75,1.85) .. controls(7.5,1) and (6.8,-0.4) .. (5.7,0);	
				
				\draw[line width=5pt, orange, opacity=0.2] (6,2.1) -- (6,3.1);
				\draw[line width=5pt, orange, opacity=0.2] (6.8,1.85) -- (7.8,2.6);
				\draw[line width=5pt, orange, opacity=0.2] (7.1,1) -- (8.1,1);
				\draw[line width=5pt, orange, opacity=0.2] (6.8,0.15) -- (7.8,-0.6);
				\draw[line width=5pt, orange, opacity=0.2] (6,-0.2) -- (6,-1.2);
				\draw[line width=5pt, orange, opacity=0.2] (4.2,-0.6) -- (5.2,0.15);
				\draw[line width=5pt, orange, opacity=0.2] (4.9,1) -- (3.9,1);
				\draw[line width=5pt, orange, opacity=0.2] (4.2,2.6) -- (5.2,1.85);
				
				\draw[line width=5pt, orange, opacity=0.2] (7.8,-0.6) -- (8.1,-1.3);
				\draw[line width=5pt, orange, opacity=0.2] (7.8,-0.6) -- (7.5,-1.3);
				
				\draw[line width=5pt, orange, opacity=0.2] (4.6,-1.3) -- (4.2,-0.6);
				\draw[line width=5pt, orange, opacity=0.2] (4.2,-1.3) -- (4.2,-0.6);
				\draw[line width=5pt, orange, opacity=0.2] (3.8,-1.3) -- (4.2,-0.6);
				
				\draw[line width=5pt, orange, opacity=0.2] (4.9,1) -- (3.9,0);
				\draw[line width=5pt, orange, opacity=0.2] (4.9,1) -- (3.9,2);
				
				\node (111) at (6,-2.1) {\textbf{(b)}};
				\label{b}
			\end{tikzpicture}
			\label{fig:case2}
		\end{subfigure}
		~
		\begin{subfigure}[t]{0.32\linewidth}
			\centering
			\begin{tikzpicture}[xscale=0.75,yscale=0.7]
				
				\draw[dashed, thick] (6,1) circle (1.1cm);
				
				\draw[dashed, thick] (6,2.1) -- (6,3.1);
				\draw[dashed, thick] (6.8,1.85) -- (7.8,2.6);
				\draw[dashed, thick] (7.1,1) -- (8.1,1);
				\draw[dashed, thick] (6.8,0.15) -- (7.8,-0.6);
				\draw[dashed, thick] (6,-0.2) -- (6,-1.2);
				\draw[dashed, thick] (4.2,-0.6) -- (5.2,0.15);
				\draw[dashed, thick] (4.9,1) -- (3.9,1);
				\draw[dashed, thick] (4.2,2.6) -- (5.2,1.85);
				
				\draw[dashed, thick] (7.8,-0.6) -- (8.1,-1.3);
				\draw[dashed, thick] (7.8,-0.6) -- (7.5,-1.3);
				
				\draw[dashed, thick] (4.6,-1.3) -- (4.2,-0.6);
				\draw[dashed, thick] (4.2,-1.3) -- (4.2,-0.6);
				\draw[dashed, thick] (3.8,-1.3) -- (4.2,-0.6);
				
				\draw[dashed, thick] (4.9,1) -- (3.9,0);
				\draw[dashed, thick] (4.9,1) -- (3.9,2);
				
				\node[label=left:{\tiny{$b_i$}}] (1) at (6.5,1.8) {};
				\node[label=right:{\tiny{$b_{j+1}$}}] (5) at (5.8,-0.3) {};
				\node[label=left:{\tiny{$b_j$}}] (6) at (5.7,-0.15) {};
				\node[label=left:{\textcolor{teal}{\tiny{$2 \times X_{k}$}}}] (6) at (6.65,0.36) {};
				
				\node at (6,2.1) [circle,draw, fill=black, opacity=1, color=black, inner sep=0.3mm] (11) {};
				\node at (6.75,1.83) [circle,draw, fill=black, opacity=1, color=black, inner sep=0.3mm] (12) {};
				\node at (7.1,1) [circle,draw, fill=black, opacity=1, color=black, inner sep=0.3mm] (13) {};
				\node at (6.75,0.17) [circle,draw, fill=black, opacity=1, color=black, inner sep=0.3mm] (14) {};
				\node at (6,-0.11) [circle,draw, fill=black, opacity=1, color=black, inner sep=0.3mm] (15) {};
				\node at (5.25,0.17) [circle,draw, fill=black, opacity=1, color=black, inner sep=0.3mm] (16) {};
				\node at (4.9,1) [circle,draw, fill=black, opacity=1, color=black, inner sep=0.3mm] (17) {};
				\node at (5.25,1.82) [circle,draw, fill=black, opacity=1, color=black, inner sep=0.3mm] (18) {};
				
				\draw[-, line width=5pt, teal, opacity=0.2] (5.44,0.06) .. controls(5.5,0.045) and (5.65,0.02) .. (5.7,0);
				
				\draw[-, line width=5pt, orange, opacity=0.2] (5.45,0.05) .. controls(5.13,0.2) and (5.05,0.4) .. (5,0.5) .. controls(4.6,2.1) and (6,2.4) .. (6.75,1.85) .. controls(7.5,1) and (6.8,-0.4) .. (5.7,0);	
				
				\draw[line width=5pt, orange, opacity=0.2] (6,2.1) -- (6,3.1);
				\draw[line width=5pt, orange, opacity=0.2] (6.8,1.85) -- (7.8,2.6);
				\draw[line width=5pt, orange, opacity=0.2] (7.1,1) -- (8.1,1);
				\draw[line width=5pt, orange, opacity=0.2] (6.8,0.15) -- (7.8,-0.6);
				\draw[line width=5pt, orange, opacity=0.2] (6,-0.2) -- (6,-1.2);
				\draw[line width=5pt, orange, opacity=0.2] (4.2,-0.6) -- (5.2,0.15);
				\draw[line width=5pt, orange, opacity=0.2] (4.9,1) -- (3.9,1);
				\draw[line width=5pt, orange, opacity=0.2] (4.2,2.6) -- (5.2,1.85);
				
				\draw[line width=5pt, orange, opacity=0.2] (7.8,-0.6) -- (8.1,-1.3);
				\draw[line width=5pt, orange, opacity=0.2] (7.8,-0.6) -- (7.5,-1.3);
				
				\draw[line width=5pt, orange, opacity=0.2] (4.6,-1.3) -- (4.2,-0.6);
				\draw[line width=5pt, orange, opacity=0.2] (4.2,-1.3) -- (4.2,-0.6);
				\draw[line width=5pt, orange, opacity=0.2] (3.8,-1.3) -- (4.2,-0.6);
				
				\draw[line width=5pt, orange, opacity=0.2] (4.9,1) -- (3.9,0);
				\draw[line width=5pt, orange, opacity=0.2] (4.9,1) -- (3.9,2);
				
				\node (111) at (6,-2.1) {\textbf{(c)}};
				\label{c}
				
			\end{tikzpicture}
			\label{fig:case3}
		\end{subfigure}
		\caption{ (a) Step $iii$ of Procedure~\ref{proc:givenCliquesRecog} where all components and the branching cliques placed on the orange subgraph of $U$ together with the teal connected components $X_k$ and $X_l$ correspond to $\mathcal{Y^*}$. (b) Step $iv$ of Procedure~\ref{proc:givenCliquesRecog} where the blue subgraph is not included in $\mathcal{Y^*}$ which contains all components and the branching cliques placed on the orange subgraph of $U$. (c) Step $v$ of Procedure~\ref{proc:givenCliquesRecog} where $\mathcal{Y^*}$ contains all components and the branching cliques placed on the orange subgraph of $U$ together with two copies of $X_k$.}
		\label{fig:cases}
	\end{figure}

	\begin{reptheorem}{givenCliquesRecog}\label{theo:givenCliquesRecog}
		Procedure~\ref{proc:givenCliquesRecog} correctly decides whether a given chordal graph is a proper $U$-graph in \emph{FPT}-time.
	\end{reptheorem}
	
	\begin{proof}
		We first prove that Procedure~\ref{proc:givenCliquesRecog} works correctly. We return that $G$ is not a proper $U$-graph if the number of rich cliques is not bounded due to Lemma~\ref{lem:richBranchBczTpartsCliques} or the number of rich cliques is less than $\vert \mathcal{B} \vert$ due to the assumption A1. Since the set of rich cliques contains the possible branching cliques by Lemma~\ref{lem:richBranchBczTpartsCliques}, we consider the assignments $f$ of rich cliques to the branching nodes $B$ on the circle of $U$. Therefore, there are at most $\vert \mathcal{B} \vert$ connected components placed on the circle in $G-f(\mathcal{C^*})$ since every edge of a proper $U$-representation carries at most one connected component \cite{SdT-graphs2021efficient} and a proper $U$-graph has at least one proper $U$-representation. Also, we recall that an assignment of maximal cliques to the branching nodes of $U$ determines the components between the branching nodes if $G$ is a proper $U$-graph \cite{SdT-graphs2021efficient}.
		
		Similarly, there are at most $deg(\mathcal{B}) - 2\vert \mathcal{B} \vert$ connected components placed not on the circle in $G-f(\mathcal{C^*})$ since every edge of the circle is counted twice. Specifically, there is exactly one connected component placed on each edge of $U$ due to the properness and the connectedness of $G$ \cite{SdT-graphs2021efficient}, and thus, for each $b_i \in \mathcal{B}$, at most $deg(b_i) -2$ connected components are attached to $C_i$ but not placed on the circle where $deg(b_i)$ denotes the degree of $b_i$. Therefore, if there are more than $deg(\mathcal{B}) - \vert \mathcal{B} \vert$ connected components in $G-f(\mathcal{C^*})$, we move on to another assignment as $f$ does not result in a proper $U$-representation.
		
		For each branching node $b_i$, we consider the defined subgraph $P$ of $U$, and obtain a subtree $Y \subsetneq U$ by including the other branching nodes $b_k$ such that $C_i \cap C_k$ is non-empty. 
		If two branching cliques of $\mathcal{C^*}$  share a vertex, then there is no connected component between them in a proper $T$-representation since the shared vertex would contain the connected component \cite{SdT-graphs2021efficient}. Therefore, if $P$ excludes some branching node on the circle as in Step $iii$, it is sufficient to check whether the components and the branching cliques contained in $Y^*$ together with the specified additional components is a proper $Y$-graph. Moreover, since $G$ is connected and $P$ excludes a branching node on the circle, both specified additional two components $X_k$ and $X_l$ exist. The branching cliques that $X_k$ is placed in between do not share a vertex, and this also holds for $X_l$. Therefore, if every such graph is a proper $Y$-graph for some tree $Y$, then $G$ is indeed a proper $U$-graph.
		
		Otherwise, $P$ is the circle and to obtain the subtree $Y \subsetneq U$, we discard a special edge of the circle. We emphasize that in cases $iv$ and $v$, we return that $G$ is a proper $U$-graph immediately when the corresponding tests are passed since the defined graph $Y^*$ contains every vertex of $G$ which may affect the properness.
		
		If there is no connected component placed on the circle as in Step $iv$, then we discard the edge between the furthest consecutive branching nodes $b_j$ and $b_{j+1}$ from $b_i$. Such nodes must exist again due to the properness we have to guarantee and since the induced subgraph of $G$ placed on $\mathcal{B}$ must be a circular-arc graph if $G$ is a proper $U$-graph. Since we remove the intersection $C_j \cap C_{j+1}$, and there exists no connected components, if $Y^*$ is a proper $Y$-graph, then $G$ is a proper $U$-graph as the intersection discarded can always be represented properly between the branching cliques containing them.
		
		Finally, if there exists a component $X_k$ between the branching nodes $b_j$ and $b_{j+1}$ as in Step $v$, then we again discard the edge $e$ between them and make two copies of $X_k$ with distinct attachments to the branching cliques $C_j$ and $C_{j+1}$. The correctness is the same as in the previous case but this time, we also need to guarantee that $X_k$ can really be represented properly. To achieve that, we additionally check whether $G[C_j \cup X_k \cup C_{j+1}]$ is a proper interval graph which is sufficient since $X_k$ has to be placed deterministically on $e$ if $G$ is a proper $U$-graph and $f$ results in a proper $U$-representation.
		
		We now prove that Procedure~\ref{proc:givenCliquesRecog} works in \emph{FPT}-time.
		The set $\mathcal{C}$ is found in polynomial time using the algorithm given in \cite{listAllMaxCliques}. Among those maximal cliques in $\mathcal{C}$, one can in \emph{FPT}-time identify the set $\mathcal{C^*}$ of rich cliques by Lemma~\ref{lem:richBranchBczTpartsCliques}. Since the number of connected components in $\mathcal{X}$ must be bounded by $\vert U \vert$ in a proper $U$-graph, if $\mathcal{C^*}$ is bounded, we can try all possible assignments $f$, and for each of them, we can find the deterministic placements of $\mathcal{X}$ on $U$ in \emph{FPT}-time. We test in \emph{FPT}-time whether every obtained $Y^*$ is a proper $Y$-graph for the tree $Y$ and the proper interval graph recognition using the algorithm given in \cite{chaplick2020recognizing}. Therefore, the overall procedure works correctly and in \emph{FPT}-time parameterized by $\vert U \vert$.
	\end{proof}
	
	Since all maximal cliques of graphs with polynomially many maximal cliques can be listed in polynomial time \cite{DBLP:journals/siamcomp/TsukiyamaIAS77}, we obtain the following. 
	
	\begin{theorem}\label{lem:polyManyRecogIso}
		For a fixed unicyclic graph $U$, one can decide whether an input graph with polynomially many maximal cliques is a proper $U$-graph in \emph{FPT}-time parameterized by $\vert U \vert$. 
	\end{theorem}

	\begin{proof}
		We first list all maximal cliques of our input graphs in polynomial time \cite{DBLP:journals/siamcomp/TsukiyamaIAS77}. We identify the rich cliques among them efficiently and their numbers for proper $U$-graphs are bounded by $\vert U \vert$ by Lemma~\ref{lem:richBranchBczTpartsCliques}. Thus, if this is not the case, one can correctly return that a given graph is not a proper $U$-graph. Now, since both the number of rich cliques and the number of branching nodes of the circle of $U$ is bounded by $\vert U \vert$, the number of all possible assignments of rich cliques to the branching nodes of the circle of $U$ is bounded by a function of $U$. For each such assignment, we use Procedure~\ref{proc:givenCliquesRecog} and if at least one assignment returns that the input graph is a proper $U$-graph, we return its output. Otherwise, all assignments fail and the input graph has no proper $U$-representation certifying that it is not a proper $U$-graph. Since Procedure~\ref{proc:givenCliquesRecog} works correctly in \emph{FPT}-time and the number of all assignments is bounded by a function of $U$, one can recognize the input graphs with polynomially many maximal cliques in \emph{FPT}-time parameterized by $\vert U \vert$.
	\end{proof}
	
	Since every chordal graph is a $T$-graph for some tree $T$~\cite{chordalityInters}, we give the following.
	
	\begin{conj}\label{con:UtoT-graph2}
		
		If $G$ is a connected chordal $U$-graph for a fixed unicyclic graph $U$, then $G$ is a $T$-graph for some tree $T$ which can be obtained from $U$ in linear time.		
	\end{conj}
	
	We now focus on the recognition problem for proper Helly $U$-graphs, i.e. proper $U$-graphs which have at least one proper $U$-representation such that all cliques are Helly. It has been shown that every (proper) Helly $H$-graph $G$ has at most $\vert H \vert + \vert E(H) \vert \cdot \vert G \vert$ maximal cliques which is linear on the order of the input graph $G$ \cite{zemanWG}. Recall that all maximal cliques of graphs with polynomially many maximal cliques can be listed in polynomial time \cite{DBLP:journals/siamcomp/TsukiyamaIAS77}. Under the same assumptions, we modify Procedure~\ref{proc:givenCliquesRecog} to decide whether a given graph $G$ is a proper Helly $U$-graphs as follows:
	
	\begin{enumerate}
		\item Before Step 1, we start listing the maximal cliques of $G$ using the algorithm of \cite{DBLP:journals/siamcomp/TsukiyamaIAS77}. If the listing procedure lists at least $\vert U \vert + \vert E(U) \vert \cdot \vert G \vert$ maximal cliques, we terminate the enumeration algorithm and return that $G$ is not a proper Helly $U$-graph.
		
		\item In Step 2, we only consider the assignments $f$ such that $G[f(\mathcal{C^*})]$ is a Helly circular-arc graph tested using the algorithm of \cite{DBLP:conf/cocoon/LinS06}.
		
		\item We return either $G$ is a proper Helly $U$-graph or not, analogously to Procedure~\ref{proc:givenCliquesRecog}.
	\end{enumerate}

	\begin{reptheorem}{easyIso}\label{theo:HellyRecog}
		The above modifications to Procedure~\ref{proc:givenCliquesRecog} result in correctly recognizing proper Helly $U$-graphs in \emph{FPT}-time.
	\end{reptheorem}

	\begin{proof}
		The first modification guarantees that the given graph $G$ has linearly many maximal cliques on the order of $G$ if it is a proper Helly $U$-graph  \cite{zemanWG}. For a tree $T$, every $T$-representation of a $T$-graph satisfies the Helly property \cite{zemanWG}. Therefore, only the cliques of a proper Helly $U$-graph placed on the cycle of $U$ can be non-Helly. Thus, the correctness follows from Theorem~\ref{lem:polyManyRecogIso}, and the first and the second modifications. Also the procedure takes \emph{FPT}-time again by Theorem~\ref{lem:polyManyRecogIso}, and since the tests given in the first and the second modifications can be checked in polynomial time \cite{DBLP:conf/cocoon/LinS06,DBLP:journals/siamcomp/TsukiyamaIAS77}.
	\end{proof}

	\section{Recognizing proper $\boldsymbol{U}$-graphs in general}
	\label{sec:expoCliques}

	Here we assume that the input graphs have exponentially many maximal cliques since otherwise we can use Procedure~\ref{proc:givenCliquesRecog} by Theorem~\ref{lem:polyManyRecogIso}. Therefore, we assume that the input graph is not a proper Helly $U$-graph and not a chordal graph. Given a graph $G$, we make use of the following observations: 
	
	\begin{itemize}
		\item If $G$ is a proper $U$-graph but not a proper Helly $U$-graph, then every proper $U$-representation of $G$ contains at least one non-Helly clique $C$. Since $T$-graphs satisfy the Helly property for a tree $T$, $C$ (and all such cliques of $G$) has to be placed on the circle in every $U$-representation if $G$ is a (proper) $U$-graph. Thus, we aim to find the cliques which are always non-Helly since we are not given a $U$-representation and being a non-Helly clique may depend on a particular $U$-representation.
		
		\item If $G$ is not chordal, then $G$ contains at least one induced cycle $L$ of length at least $4$ which is referred as a \emph{``hole''}. Since $T$-graphs are chordal, $L$ (and all holes of $G$) has to be placed on the circle in every $U$-representation if $G$ is a (proper) $U$-graph. Therefore, our aim is to identify the holes. 
		
		\item The union of all vertices placed completely outside of the circle of $U$ in a proper $U$-representation forms a (possibly disconnected) chordal graph where a vertex $v$ being placed ``completely'' on, and analogously, outside of the subgraph~$S$ of $U$ means that $v$ is represented with some subgraph of (a subdivision of)~$S$, and $U \setminus S$ giving such a proper $U$-representation, respectively.
	\end{itemize}
	
	We first give the following lemma which provides $i)$ an upper bound on the length of every hole $L$ of a proper $U$-graph which is placed completely on the circle of $U$, and $ii)$ an upper bound on the length of any hole $L$.
	
	\begin{replemma}{inducedcycle<=6andU}\label{lem:inducedcycle<=6andU}
		For a fixed unicyclic graph $U$, let $G$ be a proper $U$-graph which is not a proper Helly $U$-graph nor a chordal graph. In a proper $U$-representation of $G$, if a hole $L \subseteq G$ is placed completely on the circle of $U$, then the length of $L$ is at most $6$. Moreover, the length of any hole $L \subseteq G$ is linearly bounded in $\vert U \vert$. 
	\end{replemma}

	\begin{proof}
		Note that since $G$ is not a proper Helly $U$-graph nor a chordal graph, there exist at least one non-Helly clique and at least one hole of $G$ which have to be placed (completely or not) on the circle of $U$ in every proper $U$-representation. 
		
		We first prove that the length of every hole $L \subseteq G$ placed completely on the circle of $U$ is at most $6$. We assume that there is no vertex of $G$ covering the whole circle of $U$ since otherwise such $L$ can not be properly represented. Let $C$ denote a non-Helly clique of $G$ in this representation. Since the considered representation is proper and we aim to determine the maximum length of such a cycle $L$, each vertex of $L$ must have the smallest number of neighbors from $C$. Note that if a vertex of $L$ represented completely on the circle is adjacent to one vertex of $C$, then it is contained by that vertex of $C$ resulting in a non-proper representation. Thus, we fix every vertex of $L$ to have exactly $2$ neighbors from $C$. Observe that at most $2$ vertices of $L$ can have the same pair of neighbors from $C$ due to the properness and since $L$ is a hole. Also, since $L$ is a hole, two vertices of $L$ can have at most one common neighbor from $L$. By this construction maximizing the length of $L$, if $C$ is of size $2$, then the length of $L$ is at most $4$, and otherwise there is at least one non-Helly triangle of size $3$, and therefore the length of $L$ is at most $6$. 
		
		We now prove that the length of any hole $L \subseteq G$ is linearly bounded in $\vert U \vert$. Since $L$ is induced, no three vertices of $L$ are pairwise adjacent. Then, at most $2$ vertices of $L$ can go through the same branch on the circle of $U$ and these vertices are consecutive in $L$. By the previous arguments, the number of vertices of $L$ not going through a branch of $U$ is bounded. Thus, to obtain the maximum length hole $L$, 2 consecutive vertices of $L$ must go through a branch together and none of them should appear in other branches. Therefore, the length of any hole is bounded by $2 \vert \mathcal{B} \vert$ where $\mathcal{B} \leq \vert U \vert$ denotes the number of branching nodes of $U$ on the circle. 
	\end{proof}

	Note that in proper Helly $U$-graphs, the length of holes is not bounded since they do not contain non-Helly cliques. However, they can be recognized in \emph{FPT}-time without identifying their holes by Theorem~\ref{theo:givenCliquesRecog}. On the other hand, one can identify all holes of proper $U$-graphs (satisfying Helly property or not) by brute-force in \emph{XP}-time by testing all $d$-tuples of vertices for $4 \leq d \leq \vert U \vert$ by Lemma~\ref{lem:inducedcycle<=6andU}. Before we show how to compute the holes in \emph{FPT}-time, we give the following lemmas which help us identify the non-Helly cliques of the input graph $G$ given the holes, i.e. the cliques which are always non-Helly if $G$ is a proper $U$-graph.

	\begin{replemma}{circleVertices}\label{lem:circleVertices}
		For a fixed unicyclic graph $U$, let $G$ be a proper $U$-graph and $L$ be a hole of $G$. Every vertex $v$ of $G$ adjacent to at least $3$ vertices of $L$ must be placed (not necessarily completely) on the circle of $U$ in every (proper) $U$-representation of $G$.
	\end{replemma}
	
	\begin{proof}
		Since at most $2$ vertices of $L$ can go through the same branch on the circle of $U$ and every vertex of a $U$-graph is represented by a connected subgraph of $U$, these at least $3$ vertices of $L$ are consecutive in $L$ and any vertex $v$ of $G$ adjacent to at least three vertices of $L$ must contain two intersections of such vertices which are either on the circle or on the edges of $U$ other than the circle but incident to two consecutive vertices of $U$ on the circle. Then, the vertex $v$ must be placed on the circle but not necessarily completely meaning that a partial representation on the circle is guaranteed.
	\end{proof}

	\begin{replemma}{nonHellytogether}\label{lem:nonHellytogether}
		For a fixed unicyclic graph $U$, let $G$ be a proper $U$-graph which is not chordal. If $G$ is not a proper Helly $U$-graph, one can identify all vertices contained in non-Helly cliques of $G$, i.e. the cliques which are non-Helly in every proper $U$-representation of $G$, in polynomial time given the set $\mathcal{L}$ of all vertices contained in the holes of $G$.
	\end{replemma}

	\begin{proof}
		Since $G$ is not a proper Helly $U$-graph, there is at least one non-Helly clique $C$ which has to be placed on the circle in every proper $U$-representation of $G$. We assume that no vertex $v$ of $C$ covers the whole circle of $U$ since otherwise $C$ is a Helly clique as $v$ intersect also the intersection of its other two vertices. By the definition of Helly property, for every non-Helly clique of size at least $4$, we have a clique of size $3$ (which is a subclique of that non-Helly clique) covering the whole circle. Thus, with the upcoming method, we identify all vertices of non-Helly cliques, i.e. the cliques which are non-Helly in every proper $U$-representation of $G$, even though we do not identify all non-Helly cliques among them which can be exponentially many.

		Since $G$ is not chordal, all the following hold:
		
		\begin{itemize}
			\item $G$ contains a hole $L$.
			\item Every hole $L$ of $G$ covers the circle of $U$.
			\item Every hole $L$ of $G$ has a fixed cyclic ordering (up to reversal) around the circle of $U$ since $L$ is an induced cycle of length at least $4$.
			\item The union of the neighborhoods of the vertices of every hole $L$ of $G$ is a superset of all vertices appearing in every non-Helly clique $C$ of $G$.
		\end{itemize}
		
		Since $\vert L \vert \geq 4$, at least one vertex of $C$ has at least $2$ neighbors consecutive in $L$. Then, for at least one hole $L \in G$, at least two vertices of $C$ have incomparable (by inclusion) neighborhoods with each other in $L$ since there is a cyclic order on $L$ and all vertices of $C$ do not meet at a common point. This holds since otherwise, $C$ can be represented as a Helly clique. By checking the neighborhoods of every clique of size $3$ which union of non-empty neighborhoods on $\mathcal{L}$ equals $\mathcal{L}$ in polynomial time and comparing their neighborhoods in $\mathcal{L}$, one can identify every such clique $C$ given $\mathcal{L}$ appearing as a non-Helly clique in every proper $U$-representation of $G$.
	\end{proof}

	\begin{lemma}\label{lem:dominatingEdgesOrT-graph}
		For a fixed unicyclic graph $U$, let $G$ be a proper $U$-graph which is not chordal. If $G$ contains a dominating edge covering the circle of $U$, i.e. an edge which endvertices covers the whole circle, one can identify all dominating edges in polynomial time given the set $\mathcal{L}$ of all vertices contained in the holes of $G$.
	\end{lemma}

	\begin{proof}
		
		If none of two vertices $u$ and $v$ of such a dominating edge $e$ covers the whole circle of $U$, those vertices have incomparable (by inclusion) neighborhoods in at least one hole and can be identified in polynomial time as there are polynomially many edges. Otherwise, the vertex $u$ dominates the whole circle but may not cover it. Then, the neigborhood of $v$ in the circle is a subset of the neigborhood of $u$ in the circle but in this case, $v$ is actually a vertex of some hole of $G$, thus already given in $\mathcal{L}$. 
	\end{proof}

	We give the following procedure to obtain the set $\mathcal{L}$ of \emph{hole vertices} consisting of all vertices contained in the holes of an input graph $G$, and the possibly empty set $\mathcal{R}$ of \emph{revealed vertices} which are not in $\mathcal{L}$ but also must be represented (not necessarily completely) on the circle of $U$ if $G$ is a proper $U$-graph.
	
	\begin{proc}\label{proc:allInduced} \rm
		
		Given a connected graph $G$ which is not chordal, we identify the set $\mathcal{L}$ of hole vertices and (if exists) a set of revealed vertices $\mathcal{R}$ of $G \setminus \mathcal{L}$ as follows:
		
		\begin{enumerate}
			\item Let $\mathcal{L}$ and $\mathcal{R}$ denote the initially empty sets of hole vertices and revealed vertices of $G$, respectively. 
			\item For each vertex $u \in G$ and its each pair of neighbors $v,w$:
			\begin{itemize}
				\item If $v$ and $w$ are not adjacent, identify the shortest path $P$ between $v$ and $w$ in $G - (N[u] \setminus \{v,w\})$ using Dijkstra's algorithm \cite{DBLP:journals/nm/Dijkstra59} where $N[u]$ denotes the closed neighborhood of $v$. If $P \neq \emptyset$, $\mathcal{L} \leftarrow \mathcal{L} \cup V(P) \cup \{u,v,w\}$.
			\end{itemize}
			
			\item For each vertex $u \in G \setminus \mathcal{L}$:	
			\begin{itemize}
				\item If $u$ has at least three neighbors in $\mathcal{L}$ forming an induced path in $G$, then $\mathcal{R} \leftarrow \mathcal{R} \cup u$. This can be checked by identifying whether there exist three distinct vertices $v, w, x \in (N[u] \setminus u) \cap \mathcal{L}$ such that $\{v,w,x\}$ forms an induced path of length $2$ in $G$.
			\end{itemize} 
			\item Return $\mathcal{L}$ as the set of hole vertices and $\mathcal{R}$ as the set of revealed vertices.	
		\end{enumerate}
	\end{proc}
	
	\begin{lemma}\label{lem:revealerSet}
		Procedure~\ref{proc:allInduced} correctly identifies the set $\mathcal{L}$ of hole vertices and the set $\mathcal{R}$ of revealed vertices in polynomial time.
	\end{lemma}
	
	\begin{proof}
		The correctness and complexity of identifying the set $\mathcal{L}$ of hole vertices depends on the Dijkstra's algorithm \cite{DBLP:journals/nm/Dijkstra59}, and greedily checking the neighborhood of each vertex of $G$.
		We now prove that $\mathcal{R}$ is indeed a revealed set. It is clear that $\mathcal{R}$ can be computed in polynomial time. Since only the consecutive vertices of holes are connected by an edge, each vertex $u$ having at least three neighbors $v, w, x \in \mathcal{L}$ appearing as an induced path in $G$ has to be placed on the circle by Lemma~\ref{lem:circleVertices}. Also, if $u$ has more that three such neighbors, some triple among them still form an induced path of length $2$, and therefore, it is sufficient to check all three neighbors of $u$ in polynomial time. 
	\end{proof}
	
	Note that all holes found using Procedure~\ref{proc:allInduced} lie on the circle of $U$ in every proper $U$-representation of $G$ if $G$ is a proper $U$-graph since they are holes also in $G$ and the rest of $G$ is chordal. We now give Procedure~\ref{proc:generalRecog} to recognize general proper $U$-graphs.
	
	\begin{proc}\label{proc:generalRecog} \rm
		Given a connected graph $G$ on $n$ vertices and a fixed unicyclic graph $U$, we decide whether $G$ is a proper $U$-graph as follows:
		\begin{enumerate}
			\item If $G$ is chordal, then run Procedure~\ref{proc:givenCliquesRecog} and return its output.
			
			\item Check whether $G$ is a proper Helly $U$-graph using Procedure~\ref{proc:givenCliquesRecog} with the modifications proven in Theorem~\ref{theo:HellyRecog}. If it returns that $G$ is a proper Helly $U$-graph, return its output.
			
			\item Find the set $\mathcal{L}$ of hole vertices and the set $\mathcal{R}$ of revealed vertices using Procedure~\ref{proc:allInduced}.

			\item Identify the set $\mathcal{C}$ of all vertices contained in non-Helly cliques (i.e. the cliques which can not be represented as Helly due to their adjacencies in $\mathcal{L} \cup \mathcal{R}$) and dominating edges of $G$ using Lemma~\ref{lem:nonHellytogether} and Lemma~\ref{lem:dominatingEdgesOrT-graph} given the set $\mathcal{L} \cup \mathcal{R}$. If $G[\mathcal{L} \cup \mathcal{R} \cup \mathcal{C}]$ is not a circular-arc graph, return that $G$ is not a proper $U$-graph. \Comment{\emph{We do not check whether $G[\mathcal{L} \cup \mathcal{R} \cup \mathcal{C}]$ is a proper circular-arc graph since properness can be achieved on the branches of $U$.}}

			\item Find the set $\mathcal{X}$ of connected components of $G-(\mathcal{L} \cup \mathcal{R} \cup \mathcal{C})$. If $\mathcal{X}$ has more than $\vert U \vert$ components, return that $G$ is not a proper $U$-graph.
			
			\item Identify the upper attachments of each component of $\mathcal{X}$ in $G[\mathcal{L} \cup \mathcal{R} \cup \mathcal{C}]$ with respect to the fixed cyclic order of $G[\mathcal{L} \cup \mathcal{R}]$. If there is a connected component with at least $3$ upper attachments, return that $G$ is not a proper $U$-graph.

			\item Every connected component in $\mathcal{X}$ has at most $2$ upper attachments. Let $\mathcal{B}$ denote the branching nodes on the circle of $U$, $deg(\mathcal{B})$ be the sum of degrees of nodes in $\mathcal{B}$ and $\mathcal{F}$ be the family of all upper attachments. If there are more than $deg(\mathcal{B})$ such attachments in $\mathcal{F}$, return that $G$ is not a proper $U$-graph.

			\item \label{item:partition} To identify the Helly cliques placed on $\mathcal{B}$ (if $G$ is a proper $U$-graph), consider every partition of $\mathcal{X}$ containing at most $\vert U \vert$ components into subsets $\mathcal{X}_1, \dots, \mathcal{X}_{\vert \mathcal{B} \vert}$ such that:
			\begin{enumerate}
				\item For every pair $1 \leq i < j \leq \vert \mathcal{B} \vert$, it holds that $\mathcal{X}_i \cap \mathcal{X}_j = \emptyset$.

				\item For each $1 \leq i \leq \vert \mathcal{B} \vert$, $\mathcal{X}_i$ contains at most $deg(b_i) - 2$ connected components with exactly $1$ upper attachment where $deg(b_i)$ denotes the degree of $b_i$.
				
				\item Let $\mathcal{F}_1, \dots, \mathcal{F}_{\vert \mathcal{B} \vert}$ be the subfamilies of $\mathcal{F}$ such that for every $1 \leq i \leq \vert \mathcal{B} \vert$,
				$\mathcal{F}_i$ consists of exactly $1$ upper attachment of each component $X_j \in \mathcal{X}_i$ with respect to the fixed cyclic order of $G[\mathcal{L} \cup \mathcal{R}]$. 
				\Comment{There are at most $4^{\vert \mathcal{B} \vert}$ distinct choices of such subfamilies.}

				\item For each $b_i \in \mathcal{B}$, it holds that $\vert \mathcal{X}_i \vert = \vert \mathcal{F}_i \vert \leq deg(b_i)$.
				
				\item The union $C_{\mathcal{F}_i}$ of all vertices of $G$ contained in each $\mathcal{F}_i$ corresponds to a  clique in $G$.
				
				\item The circular-arc graph $G[\mathcal{L} \cup \mathcal{R} \cup \mathcal{C}]$ has a circular-ones property with $C_{\mathcal{F}_1}, \dots, C_{\mathcal{F}_{\vert \mathcal{B} \vert}}$ appearing in this order.

				\item For every pair $1 \leq i \neq j \leq \vert \mathcal{B} \vert$, there exists no $X_k \in \mathcal{X}_i$ such that any of the (at most $2$) upper attachments of $X_k$ has an intersection with $C_{\mathcal{F}_j}$ which is a strict superset of the intersection of the upper attachment of some $X_l \in \mathcal{X}_j$.
			\end{enumerate}

			\item \label{item:partSatisfied} For each partition satisfying all prescribed conditions and for each choice of subfamilies of $\mathcal{F}$, use Procedure~\ref{proc:givenCliquesRecog} from Step $(a)$ with $f(\mathcal{C}^*) = \{ C_{\mathcal{F}_1}, \dots, C_{\mathcal{F}_{\vert \mathcal{B} \vert}} \}$ and the connected components $\mathcal{X}$ of $G-f(\mathcal{C^*})$.
			\Comment{\emph{Here, we may additionally identify a set of maximal Helly-cliques containing the corresponding Helly-cliques by Lemma~\ref{lem:maximalHelly}.}}
			\begin{enumerate}
				\item If Procedure~\ref{proc:givenCliquesRecog} returns that $G$ is a proper $U$-graph, return its output. 
			\end{enumerate}
			\item Return that $G$ is not a proper $U$-graph.
			
		\end{enumerate}
		
	\end{proc}

	We first give the following lemmas regarding Steps \ref{item:partition} and \ref{item:partSatisfied} of Procedure~\ref{proc:generalRecog} which follow from the bound on the number of connected components. 
	
	\begin{lemma}\label{step7inFPT}
		The number of partitions of $\mathcal{X}$ into $\mathcal{X}_1, \dots, \mathcal{X}_{\vert \mathcal{B} \vert}$ computed in Step \ref{item:partition} of Procedure~\ref{proc:generalRecog} is bounded by a function of the parameter $\vert U \vert$, and all partitions satisfying the given conditions can be computed in \emph{FPT}-time.
	\end{lemma}

	\begin{proof}
		Since $\mathcal{X}$ has at most $\vert U \vert$ connected components and $\vert \mathcal{B} \vert < \vert U \vert$, the number of distinct partitions of $\mathcal{X}$ into $\mathcal{X}_1, \dots, \mathcal{X}_{\vert \mathcal{B} \vert}$ is upper bounded by ${\vert U \vert} ^{\vert \mathcal{B} \vert} \leq {\vert U \vert} ^{\vert U \vert}$. In addition, the number of upper attachments is bounded by $2 \vert U \vert$ since each connected component has at most two upper attachments. Thus, for each partition, there are at most $4^{\vert \mathcal{B} \vert}$ distinct choices of such subfamilies as each $\mathcal{X}_i$ contains at most two connected components with two upper attachments. Overall, we have a bounded number of partitions $\mathcal{X}$ and subfamilies of $\mathcal{F}$. We can in polynomial time check whether a set of vertices corresponds to a clique of $G$, and a given subset $C_{\mathcal{F}_1}, \dots, C_{\mathcal{F}_{\vert \mathcal{B} \vert}}$ of cliques placed on the circle can appear in this order using the PC-tree of the circular-arc graph $G[\mathcal{L} \cup \mathcal{R} \cup \mathcal{C}]$ as in \cite{DBLP:journals/dmtcs/CurtisLMNSSS13}. We also note here that this check is not the same as partial representation extension with predefined sets $C_{\mathcal{F}_1}, \dots, C_{\mathcal{F}_{\vert \mathcal{B} \vert}}$ which has been shown to be NP-complete \cite{DBLP:journals/corr/abs-2108-13076} since we do not specify the endpoints of the arcs for the corresponding vertices.
	\end{proof}
	
	\begin{lemma}\label{lem:maximalHelly}
		Considering the setting in Step \ref{item:partition} of Procedure~\ref{proc:generalRecog}, one can in \emph{FPT}-time identify a set of maximal Helly-cliques to be used in Procedure~\ref{proc:generalRecog} as commented in Step \ref{item:partSatisfied} if such a representation exists.
	\end{lemma}

	\begin{proof}
		
		By Lemma~\ref{step7inFPT}, we have a bounded number of partitions $\mathcal{X}_1, \dots, \mathcal{X}_{\vert \mathcal{B} \vert}$ and by the conditions $7.d$, $7.e$ and $7.g$ of Procedure~\ref{proc:generalRecog}, for each $1 \leq i \leq j \leq \vert \mathcal{B} \vert$, any $X_k \in \mathcal{X}_i$ and its upper attachment contains a maximal clique $C$ of $G$ which can not be obtained from any other $\mathcal{X}_j$ for $j \neq i$ as the connected components in $\mathcal{X}$ are disjoint. Since any set $C_{\mathcal{F}_1}, \dots, C_{\mathcal{F}_{\vert \mathcal{B} \vert}}$ of cliques satisfying all conditions of Step \ref{item:partition} has a circular-ones ordering, any non-maximal $C_{\mathcal{F}_i}$ can be extended to a maximal clique if there exists a component $X_j \in \mathcal{X}_i$ such that the upper attachment of $X_j$ is a superset of the upper attachment of every other component in $C_{\mathcal{F}_i}$. Otherwise, either $C_{\mathcal{F}_i}$ is already a maximal clique, or its vertices appear in some Helly-clique and it can not be extended to a maximal clique to be placed on a branch. On the other hand, every clique in $C_{\mathcal{F}_1}, \dots, C_{\mathcal{F}_{\vert \mathcal{B} \vert}}$ which can be extended to a maximal clique can be computed in \emph{FPT}-time and used in Procedure~\ref{proc:givenCliquesRecog} from Step $(a)$ as the current bijection $f(\mathcal{C}^*)$. 	
	\end{proof}

	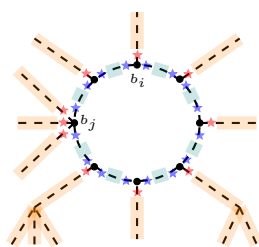
\begin{figure}[t]
		\centering
		\begin{subfigure}[t]{0.32\linewidth}
			\centering
			\begin{tikzpicture}[xscale=0.75,yscale=0.7]
				
				\draw[dashed, thick] (6,1) circle (1.1cm);
				
				\draw[dashed, thick] (6,2.1) -- (6,3.1);
				\draw[dashed, thick] (6.8,1.85) -- (7.8,2.6);
				\draw[dashed, thick] (7.1,1) -- (8.1,1);
				\draw[dashed, thick] (6.8,0.15) -- (7.8,-0.6);
				\draw[dashed, thick] (6,-0.2) -- (6,-1.2);
				\draw[dashed, thick] (4.2,-0.6) -- (5.2,0.15);
				\draw[dashed, thick] (4.9,1) -- (3.9,1);
				\draw[dashed, thick] (4.2,2.6) -- (5.2,1.85);
				
				\draw[dashed, thick] (7.8,-0.6) -- (8.1,-1.3);
				\draw[dashed, thick] (7.8,-0.6) -- (7.5,-1.3);
				
				\draw[dashed, thick] (4.6,-1.3) -- (4.2,-0.6);
				\draw[dashed, thick] (4.2,-1.3) -- (4.2,-0.6);
				\draw[dashed, thick] (3.8,-1.3) -- (4.2,-0.6);
				
				\draw[dashed, thick] (4.9,1) -- (3.9,0);
				\draw[dashed, thick] (4.9,1) -- (3.9,2);
				
				\node[label=left:{\tiny{$b_i$}}] (1) at (6.5,1.8) {};
				\node[label=left:{\tiny{$b_j$}}] (6) at (5.65,1) {};
				
				\node at (6,2.1) [circle,draw, fill=black, opacity=1, color=black, inner sep=0.3mm] (11) {};
				\node at (6.75,1.83) [circle,draw, fill=black, opacity=1, color=black, inner sep=0.3mm] (12) {};
				\node at (7.1,1) [circle,draw, fill=black, opacity=1, color=black, inner sep=0.3mm] (13) {};
				\node at (6.75,0.17) [circle,draw, fill=black, opacity=1, color=black, inner sep=0.3mm] (14) {};
				\node at (6,-0.11) [circle,draw, fill=black, opacity=1, color=black, inner sep=0.3mm] (15) {};
				\node at (5.25,0.17) [circle,draw, fill=black, opacity=1, color=black, inner sep=0.3mm] (16) {};
				\node at (4.9,1) [circle,draw, fill=black, opacity=1, color=black, inner sep=0.3mm] (17) {};
				\node at (5.25,1.82) [circle,draw, fill=black, opacity=1, color=black, inner sep=0.3mm] (18) {};
				
				\draw[line width=5pt, orange, opacity=0.2] (6,2.4) -- (6,3.1);
				\node[star, star points=5, star point ratio=2.25, draw=red, opacity=0.5, fill=red, inner sep=0.5pt, anchor=outer point 3] at (5.95,2.2) {};
				
				\draw[line width=5pt, orange, opacity=0.2] (7,2) -- (7.8,2.6);
				\node[star, star points=5, star point ratio=2.25, draw=red, opacity=0.5, fill=red, inner sep=0.5pt, anchor=outer point 3] at (6.84,1.84) {};
				
				\draw[line width=5pt, orange, opacity=0.2] (7.4,1) -- (8.1,1);
				\node[star, star points=5, star point ratio=2.25, draw=red, opacity=0.5, fill=red, inner sep=0.5pt, anchor=outer point 3] at (7.24,0.92) {};
				
				\draw[line width=5pt, orange, opacity=0.2] (7,0) -- (7.8,-0.6);
				\node[star, star points=5, star point ratio=2.25, draw=red, opacity=0.5, fill=red, inner sep=0.5pt, anchor=outer point 3] at (6.85,0) {};
				
				\draw[line width=5pt, orange, opacity=0.2] (6,-0.4) -- (6,-1.2);
				\node[star, star points=5, star point ratio=2.25, draw=red, opacity=0.5, fill=red, inner sep=0.5pt, anchor=outer point 3] at (5.95,-0.4) {};
				
				\draw[line width=5pt, orange, opacity=0.2] (4.2,-0.6) -- (5,0);
				\node[star, star points=5, star point ratio=2.25, draw=red, opacity=0.5, fill=red, inner sep=0.5pt, anchor=outer point 3] at (5.05,0) {};
				
				\draw[line width=5pt, orange, opacity=0.2] (4.6,1) -- (3.9,1);
				\node[star, star points=5, star point ratio=2.25, draw=red, opacity=0.5, fill=red, inner sep=0.5pt, anchor=outer point 3] at (4.65,0.925) {};
				
				\draw[line width=5pt, orange, opacity=0.2] (4.2,2.6) -- (5,2);
				\node[star, star points=5, star point ratio=2.25, draw=red, opacity=0.5, fill=red, inner sep=0.5pt, anchor=outer point 3] at (5.05,1.85) {};
				
				\draw[line width=5pt, orange, opacity=0.2] (7.8,-0.6) -- (8.1,-1.3);
				\draw[line width=5pt, orange, opacity=0.2] (7.8,-0.6) -- (7.5,-1.3);
				
				\draw[line width=5pt, orange, opacity=0.2] (4.6,-1.3) -- (4.2,-0.6);
				\draw[line width=5pt, orange, opacity=0.2] (4.2,-1.3) -- (4.2,-0.6);
				\draw[line width=5pt, orange, opacity=0.2] (3.8,-1.3) -- (4.2,-0.6);
				
				\draw[line width=5pt, orange, opacity=0.2] (4.6,0.73) -- (3.9,0);
				\node[star, star points=5, star point ratio=2.25, draw=red, opacity=0.5, fill=red, inner sep=0.5pt, anchor=outer point 3] at (4.65,1.15) {};
				
				\draw[line width=5pt, orange, opacity=0.2] (4.6,1.27) -- (3.9,2);
				\node[star, star points=5, star point ratio=2.25, draw=red, opacity=0.5, fill=red, inner sep=0.5pt, anchor=outer point 3] at (4.65,0.7) {};
				
				\draw[-, line width=5pt, teal, opacity=0.2] (6.53,0.04) -- (6.28,-0.04);
				\node[star, star points=5, star point ratio=2.25, draw=blue, opacity=0.5, fill=blue, inner sep=0.5pt, anchor=outer point 3] at (6.57,-0.04) {};
				\node[star, star points=5, star point ratio=2.25, draw=blue, opacity=0.5, fill=blue, inner sep=0.5pt, anchor=outer point 3] at (6.125,-0.18) {};
				
				\draw[-, line width=5pt, teal, opacity=0.2] (7.05,0.63) -- (6.95,0.4);
				\node[star, star points=5, star point ratio=2.25, draw=blue, opacity=0.5, fill=blue, inner sep=0.5pt, anchor=outer point 3] at (6.81,0.22) {};
				\node[star, star points=5, star point ratio=2.25, draw=blue, opacity=0.5, fill=blue, inner sep=0.5pt, anchor=outer point 3] at (7.04,0.68) {};
				
				\draw[-, line width=5pt, teal, opacity=0.2] (7,1.4) -- (6.9,1.63);
				\node[star, star points=5, star point ratio=2.25, draw=blue, opacity=0.5, fill=blue, inner sep=0.5pt, anchor=outer point 3] at (6.8,1.62) {};
				\node[star, star points=5, star point ratio=2.25, draw=blue, opacity=0.5, fill=blue, inner sep=0.5pt, anchor=outer point 3] at (7,1.2) {};
				
				\draw[-, line width=5pt, teal, opacity=0.2] (6.53,1.96) -- (6.28,2.04);
				\node[star, star points=5, star point ratio=2.25, draw=blue, opacity=0.5, fill=blue, inner sep=0.5pt, anchor=outer point 3] at (6.565,1.84) {};
				\node[star, star points=5, star point ratio=2.25, draw=blue, opacity=0.5, fill=blue, inner sep=0.5pt, anchor=outer point 3] at (6.1,2) {};
				
				\draw[-, line width=5pt, teal, opacity=0.2] (5.47,0.04) -- (5.72,-0.04);
				\node[star, star points=5, star point ratio=2.25, draw=blue, opacity=0.5, fill=blue, inner sep=0.5pt, anchor=outer point 3] at (5.32,-0.04) {};
				\node[star, star points=5, star point ratio=2.25, draw=blue, opacity=0.5, fill=blue, inner sep=0.5pt, anchor=outer point 3] at (5.75,-0.18) {};
				
				\draw[-, line width=5pt, teal, opacity=0.2] (5,0.6) -- (5.1,0.37);
				\node[star, star points=5, star point ratio=2.25, draw=blue, opacity=0.5, fill=blue, inner sep=0.5pt, anchor=outer point 3] at (5.07,0.19) {};
				\node[star, star points=5, star point ratio=2.25, draw=blue, opacity=0.5, fill=blue, inner sep=0.5pt, anchor=outer point 3] at (4.865,0.68) {};
				
				\draw[-, line width=5pt, teal, opacity=0.2] (4.95,1.4) -- (5.05,1.63);
				\node[star, star points=5, star point ratio=2.25, draw=blue, opacity=0.5, fill=blue, inner sep=0.5pt, anchor=outer point 3] at (5.07,1.62) {};
				\node[star, star points=5, star point ratio=2.25, draw=blue, opacity=0.5, fill=blue, inner sep=0.5pt, anchor=outer point 3] at (4.865,1.2) {};
				
				\draw[-, line width=5pt, teal, opacity=0.2] (5.47,1.96) -- (5.72,2.04);
				\node[star, star points=5, star point ratio=2.25, draw=blue, opacity=0.5, fill=blue, inner sep=0.5pt, anchor=outer point 3] at (5.32,1.84) {};
				\node[star, star points=5, star point ratio=2.25, draw=blue, opacity=0.5, fill=blue, inner sep=0.5pt, anchor=outer point 3] at (5.75,2) {};	
				
				\node (111) at (6,-2.1) {\textbf{(a)}};
				\label{a}
			\end{tikzpicture}
			
			\label{fig:hell}
		\end{subfigure}
		
		\caption{(a) After the hole vertices, the revealed vertices, non-Helly cliques and dominating edges are identified using Procedure~\ref{proc:allInduced}, the teal and the orange connected components are found without deterministically knowing their specific placements on $U$. Each teal component has two starred blue upper attachments and must be placed on the circle of $U$, and each orange component has one red starred upper attachment and must not be placed on the circle of $U$ on $\mathcal{L} \cup \mathcal{R} \cup \mathcal{C}$.}
		\label{fig:hellyCliques}
	\end{figure}
	
	\begin{theorem}\label{theo:generalRecog}
		Procedure~\ref{proc:generalRecog} correctly decides whether a given graph is a proper $U$-graph in \emph{FPT}-time parameterized by $\vert U \vert$.
	\end{theorem}

	\begin{proof}
		We first prove that Procedure~\ref{proc:generalRecog} works correctly. The correctness of Step $1$ to $4$ follows from Theorem~\ref{theo:givenCliquesRecog}, Theorem~\ref{theo:HellyRecog}, Lemma~\ref{lem:nonHellytogether} and Lemma~\ref{lem:revealerSet}. In Step $5$, if there are more than $\vert U \vert$ connected components, then $G$ can not be a proper $U$-graph since there exists no proper $U$-representation in this case \cite{SdT-graphs2021efficient}. 
		
		In Step $6$, the cyclic order of $G[\mathcal{L} \cup \mathcal{R}]$ is fixed (up to reversal) since it only contains the vertices of holes and some additional vertices which have to be represented on the circle by Lemma~\ref{lem:revealerSet}, and every component must have at most $2$ upper attachments since they are either placed on the circle of $U$ having exactly $2$ distinct upper attachments due to the existence of a proper $U$-representation, or placed not on the circle of $U$ having exactly $1$ upper attachment due to the connectedness under our assumptions if $G$ is a proper $U$-graph. Then, in Step $7$, there are at most $\vert \mathcal{B} \vert$ connected components placed on the circle of $U$ with exactly two distinct upper attachments, and at most $deg(\mathcal{B}) -2 \vert \mathcal{B} \vert$ connected components with one upper attachment placed not on the circle of $U$. Thus, if there are more than $2\vert \mathcal{B} \vert + deg(\mathcal{B}) -2 \vert \mathcal{B} \vert = deg(\mathcal{B})$ upper attachments, $G$ can not be a proper $U$-graph.
		
		In Step $8$, even though we do not know the exact placements of the connected components on the edges of $U$, we try all possible placements. Due to the necessity of having a proper $U$-representation if $G$ is a proper $U$-graph, the intersections of the upper attachments of the connected components placed on incident edges with a branching clique must be maximal, and since $G[\mathcal{L} \cup \mathcal{R} \cup \mathcal{C}]$ is a circular-arc graph because of the prescribed conditions, the branching cliques $C_{\mathcal{F}_1}, \dots, C_{\mathcal{F}_{\vert \mathcal{B} \vert}}$ have a circular-ones ordering consistent with $G[\mathcal{L} \cup \mathcal{R} \cup \mathcal{C}]$.
		
		In Step $9$, we try each such partition as at least one of them passes the test of Procedure~\ref{proc:givenCliquesRecog} if $G$ is a proper $U$-graph using the identified Helly cliques. In Step $10$, no partition gives a proper $U$-representation meaning that $G$ is not a proper $U$-graph.
		
		We now prove that Procedure~\ref{proc:generalRecog} works in \emph{FPT}-time. Steps $1$ to $4$ takes \emph{FPT}-time due to Theorem~\ref{theo:givenCliquesRecog}, Theorem~\ref{theo:HellyRecog}, Lemma~\ref{lem:nonHellytogether} and Lemma~\ref{lem:revealerSet}, and Steps $5$, $6$ and $7$ takes trivially polynomial time. In Step $8$, the number of all possible partitions of $\mathcal{X}$ into $\mathcal{X}_1, \dots, \mathcal{X}_{\vert \mathcal{B} \vert}$ is bounded by our parameter $\vert U \vert$ by Lemma~\ref{step7inFPT} and the number of possible choices of subfamilies $\mathcal{F}_1, \dots, \mathcal{F}_{\vert \mathcal{B} \vert}$ of $\mathcal{F}$ is upper bounded by $4^{\vert \mathcal{B} \vert}$ since each $\mathcal{F}_i$ contains at most two connected components with two upper attachments. Finally, we use the \emph{FPT}-approach given as Procedure~\ref{proc:givenCliquesRecog} for each of boundedly many partitions of $\mathcal{X}$ and boundedly many choices of subfamilies of $\mathcal{F}$ in Step $9$ and use its result. Therefore, the overall procedure works correctly and in \emph{FPT}-time.	
	\end{proof}

	\section{Isomorphism testing for proper $\boldsymbol{U}$-graphs}
	\label{sec:isomorphism}
	
	In this section, we show how to test the isomorphism of two proper $U$-graphs in \emph{FPT}-time parameterized by $\vert U \vert$. We modify Procedure~\ref{proc:givenCliquesRecog} and Procedure~\ref{proc:generalRecog}, and have the previous assumptions given on the corresponding sections. Here, we assume connectedness also since the isomorphism of each pair of components of given disconnected graphs can be checked pairwise. We first modify Procedure~\ref{proc:givenCliquesRecog} as follows:
	
	\begin{enumerate}
		\item On the input, we are given the fixed unicyclic graph $U$, and two connected proper $U$-graphs $G$ and $H$ on $n$ vertices which are both proper Helly $U$-graphs or chordal graphs. Let $\mathcal{C^*}$ and $\mathcal{D^*}$ denote the isomorphism invariant sets of rich cliques of $G$ and $H$ by Lemma~\ref{lem:richBranchBczTpartsCliques}, respectively.

		\item In Step 1, we check whether $\vert \mathcal{C^*} \vert = \vert \mathcal{D^*} \vert$ holds.

		\item Using Procedure~\ref{proc:givenCliquesRecog}, we fix a proper $U$-representation of $G$ where $f$ denotes the assignment of branching cliques on the circle in this fixed representation. Let $C_i$ denote the maximal clique of $\mathcal{C}^*$ placed on the branching node $b_i \in \mathcal{B}$ in the current assignment $f$, i.e.~$C_i = f{\mid_{b_i}}$. 
		
		\item We denote each assignment of $\mathcal{D^*}$ to $\mathcal{B}$ by $f'$, and only consider the assignments $f'$ which match to $f$ as being in the case of Step $iii$, $iv$ and $v$ of Procedure~\ref{proc:givenCliquesRecog}. Therefore, $P$ is the same for both graphs in every step. Let $D_i$ denote the maximal clique of $\mathcal{D}^*$ placed on the branching node $b_i \in \mathcal{B}$ in the current assignment $f'$, i.e. $D_i = f'{\mid_{b_i}}$. Here, we require also that $\vert C_i \vert = \vert D_i \vert$ for every $b_i \in \mathcal{B}$. 
		
		\item For each branching node $b_i$ of $U$, we compute the subtree $Y \subsetneq U$ and the graph $\mathcal{Y^*}$ obtained by $G$ defined as in Procedure~\ref{proc:givenCliquesRecog}. We denote by $Z \subsetneq U$ and $\mathcal{Z^*}$ the corresponding subtree for $H$ and the graph obtained by $H$ in the current assignment $f'$, respectively.
		
		\item We first check whether $Y \simeq Z$, and instead of checking for proper $Y$-graph recognition, we test whether $G[\mathcal{Y^*}] \simeq H[\mathcal{Z^*}]$ using the proper $Y$-graph isomorphism testing for the tree $Y \simeq Z$ given in \cite{SdT-graphs2021efficient}.
		
		\item In Step $iv$, we additionally check whether the intersections $C_j \cap C_{j+1}$ and $D_j \cap D_{j+1}$ have the same cardinality.
		
		\item In Step $v$, we check whether the defined interval graphs are isomorphic proper interval graphs using the algorithm given in \cite{recogIntervalLinear}.
		
		\item We return either $G \simeq H$ or $G \not\simeq H$ analogous to Procedure~\ref{proc:givenCliquesRecog}.
	\end{enumerate}
	
	\begin{reptheorem}{easyIso}\label{theo:givenCliquesIso}
		The above modifications to Procedure~\ref{proc:givenCliquesRecog} result in a correct \emph{FPT}-time isomorphism testing for given two proper $U$-graphs which are both proper Helly $U$-graphs or chordal graphs.
	\end{reptheorem}
	
	\begin{proof}
		We prove that our modifications correctly cover the cases in the Steps $iii$, $iv$ and $v$ of Procedure~\ref{proc:givenCliquesRecog}. Precisely, we show that any attachment ignored in one step of the modified procedure is checked in some other step, thus the neighborhoods of bijected vertices can also be bijected to each other.
		
		In Step $iii$, $Y$ and $Z$ do not contain at least one branching node $b_j$. Therefore, $C_k \cap G[\mathcal{Y^*}] = \emptyset$ (and $D_k \cap H[\mathcal{Z^*}] = \emptyset$) and the attachment of $C_k$ in $G[\mathcal{Y^*}]$ (and $D_k$ in $H[\mathcal{Z^*}]$) is ignored. Due to our assumptions, $C_k$ (and $D_k$) has a non-empty attachment in $G[\mathcal{Y^*}]$ (and $H[\mathcal{Z^*}]$) which results in having $b_k$ together with $b_{k-1}$ and $b_{k+1}$, and possibly some other branching nodes of the circle, in at least one other $Y' \neq Y$ (and $Z' \neq Z$) even though $Y' \simeq Y$ (and $Z' \simeq Z$) may hold, thus the attachments of any $C_k$ (and $D_k$) ignored in this step are checked in another step.
		
		In Step $iv$, there is no connected components on the edges of $P$ and we consider the furthest branching nodes $b_j$ and $b_{j+1}$ from $b_i$. There is no ignored attachment in this case and since there is also no connected component between $b_i$ and $b_{j+1}$, having the same cardinality $\vert C_j \cap C_{j+1}\vert = \vert D_j \cap D_{j+1}\vert$ is sufficient.
		
		In Step $v$, there is exactly one connected component on the edge $e$ of $P$. Similar to the previous case, there is no ignored attachments and the proper interval graph isomorphism testing for the copied connected components and their immediate attachments guarantees a bijection between the same vertices.
		
		Since the proper $T$-graph isomorphism for a tree $T$ takes \emph{FPT}-time~\cite{SdT-graphs2021efficient} and proper interval graph isomorphism takes linear time~\cite{recogIntervalLinear}, the overall proper $U$-graph isomorphism approach for given two proper Helly $U$-graphs or two chordal graphs works correctly in \emph{FPT}-time.
	\end{proof}
	
	By analogous modifications to Procedure~\ref{proc:generalRecog}, we get the following.

	\begin{reptheorem}{generalIso}\label{theo:generalIso}
		The isomorphism problem for proper $U$-graphs can be solved in \emph{FPT}-time parameterized by $\vert U \vert$.
	\end{reptheorem}

	\begin{proof}
		
		We first give the following modifications to Procedure~\ref{proc:generalRecog} for general proper $U$-graph isomorphism in addition to the modifications applied to Procedure~\ref{proc:givenCliquesRecog}:
		
		\begin{enumerate}
			
			\item We denote the sets of hole vertices of $G$ and $H$ by $\mathcal{L}$ and $\mathcal{K}$, and the sets of revealed vertices of $G$ and $H$ by $\mathcal{R}$ and $\mathcal{P}$, respectively.
			
			\item We denote the non-Helly cliques of $G$ and $H$ by $\mathcal{C}$ and $\mathcal{D}$, respectively.
			
			\item We additionally check whether $G[\mathcal{L} \cup \mathcal{R} \cup \mathcal{C}] \simeq H[\mathcal{K} \cup \mathcal{P} \cup \mathcal{D}]$ using the algorithm of \cite{DBLP:journals/corr/abs-1904-04501}. We denote the sets of connected components of $G-(\mathcal{L} \cup \mathcal{R} \cup \mathcal{C})$ and $H-(\mathcal{K} \cup \mathcal{P} \cup \mathcal{D})$ by $\mathcal{X}$ and $\mathcal{X'}$, respectively.
			
			\item We denote the families of all upper attachments of $\mathcal{X}$ and $\mathcal{X'}$ by $\mathcal{F}^{\mathcal{X}}$ and $\mathcal{F}^{\mathcal{X'}}$, respectively.  We additionally check whether the number of upper attachments in $\mathcal{F}^{\mathcal{X}}$ and $\mathcal{F}^{\mathcal{X'}}$ are the same.
			
			\item We denote the subsets of the partitionings of $\mathcal{X}$ and $\mathcal{X'}$ by $\mathcal{X}_1, \dots, \mathcal{X}_s$ and $\mathcal{X'}_1, \dots, \mathcal{X'}_t$, respectively. We additionally check whether $s = t$ and $\vert \mathcal{X}_i \vert = \vert \mathcal{X'}_i \vert$ for all $1 \leq i \leq s$.
			
			\item We denote the corresponding subfamilies of $\mathcal{F}^{\mathcal{X}}$ and $\mathcal{F}^{\mathcal{X'}}$ by $\mathcal{F}^{\mathcal{X}}_1, \dots, \mathcal{F}^{\mathcal{X}}_s$ and $\mathcal{F}^{\mathcal{X'}}_1, \dots, \mathcal{F}^{\mathcal{X'}}_t$, respectively. We additionally check whether $\vert \mathcal{F}^{\mathcal{X}}_i \vert = \vert \mathcal{F}^{\mathcal{X'}}_i \vert$ for all $1 \leq i \leq s$.
			
			\item We denote the unions of all vertices of $G$ and $H$ contained in each $\mathcal{F}^{\mathcal{X}}_i$ and $\mathcal{F}^{\mathcal{X'}}_j$ by $C_{\mathcal{F}^{\mathcal{X}}_i}$ and $D_{\mathcal{F}^{\mathcal{X'}}_j}$, respectively.
			
		\end{enumerate}
		
		We now prove the correctness and the complexity. By  Lemma~\ref{lem:nonHellytogether}, Lemma~\ref{lem:dominatingEdgesOrT-graph} and Lemma~\ref{lem:revealerSet} showing all vertices contained in the holes and non-Helly cliques can be identified in \emph{FPT}-time, Lemma~\ref{step7inFPT} showing the number of all possible partitions of $\mathcal{X}$ into $\mathcal{X}_1, \dots, \mathcal{X}_{\vert \mathcal{B} \vert}$ (and, therefore, $\mathcal{X'}$ into $\mathcal{X'}_1, \dots, \mathcal{X'}_{\vert \mathcal{B} \vert}$) is bounded by a function of the parameter $\vert U \vert$, by Theorem~\ref{theo:generalRecog} showing that the Procedure~\ref{proc:givenCliquesRecog} can be extended to recognize proper $U$-graphs in general, and by Lemma~\ref{theo:givenCliquesIso} showing that the given modifications to Procedure~\ref{proc:givenCliquesRecog} results in a correct \emph{FPT}-time isomorphism testing for proper $U$-graphs which are both proper Helly $U$-graphs or chordal graphs,
		the isomorphism problem for proper $U$-graphs in general can be tested in \emph{FPT}-time.
		
		We also note that we can use Lemma~\ref{lem:maximalHelly} to identify every set of maximal Helly-cliques possible to be used in isomorphism testing. Since the number of connected components is bounded in $G$ and $H$, here we get a bounded number of choices for such sets.
	\end{proof}

	\section{GI-completeness for $\boldsymbol{H}$-graphs and proper $\boldsymbol{H}$-graphs}
	\label{sec:GI-completeness}
	
	In this section, we prove that if $H$ is not a unicyclic graph, then the isomorphism problem is \emph{GI-complete} on proper $H$-graphs, i.e., as hard as the general graph isomorphism problem.
	Note that, for every $H$, the class of proper $H$-graphs is a subclass of $H$-graphs.
	Thus, this implies the GI-completeness of the isomorphism problem for $H$-graphs as well. 
	
	The key is the following theorem where $B$ denotes the $5$-vertex graph consisting of two triangles with one vertex identified.
	
	\begin{theorem}
		\label{thm:8-graphs_gic}
		The isomorphism problem for the class of all proper $B$-graphs is GI-complete.
	\end{theorem}
	
	Note that if a graph $H_1$ is a minor of a graph $H_2$, then every $H_1$-graph is an $H_2$-graph.
	This immediately gives the following.
	
	\begin{corollary}
		If $H$ is not unicyclic, then the isomorphism problem for the class of all proper $H$-graphs and the class of all $H$-graphs is GI-complete.
	\end{corollary}
	
	To prove Theorem~\ref{thm:8-graphs_gic}, we first give several lemmas.

	\begin{lemma}
		\label{lem:complement_of_stars}
		If $G$ is a disjoint union of stars, then the complement of $G$ is a proper circular-arc graph.
	\end{lemma}

	\begin{proof}
		Let $X\subseteq V(G)$ be the vertices of degree $1$ in $G$, $N_G(v)$ be the neighborhood of $v \in V(G)$, $\overline{G}$ be the complement of $G$, and let $Y = V(G)\setminus X$.
		For every vertex $v \in Y$, $N_G(v)$ forms an equivalence class of true twins in $\overline{G}$ where true twins are a set of vertices with the same neighborhood which are pairwise adjacent.
		The quotient graph of $\overline{G}$ with respect to this equivalence relation is exactly the complement of a matching $K_{2n} - nK_2$, which is a well-known proper circular-arc graph.
		Therefore, $\overline{G}$ is also a proper circular-arc graph.
	\end{proof}
	
	Let $G = (V, E)$ be any connected graph.
	We construct a new graph $G'$ in several steps as follows.
	First, let $G_1 = (V \cup V_1, E')$ be the graph resulting from subdividing each edge of $G$ and  let $G_2 = (V\cup V_1\cup V_2, E'')$ be the graph resulting from subdividing each edge of $G_1$.
	Further, let $G_3$ be the graph obtained from $G_2$ by adding all the edges between $V$ and $V_1$.
	Finally, we set the graph $G'$ to be the complement of $G_3$.
	
	\begin{lemma}
		\label{lem:8-graph}
		Let $G$ be a connected graph.
		Then, $G'$ is a proper $B$-graph.
	\end{lemma}

	\begin{proof}
		First note that the induced subgraphs $G_2[V\cup V_2]$ and $G_2[V_1\cup V_2]$ are disjoint unions of stars.
		Thus, by Lemma~\ref{lem:complement_of_stars}, the induced subgraphs $G'[V\cup V_2]$ and $G'[V_1\cup V_2]$ are circular-arc graphs with representations $R_1\colon V\cup V_2 \to \mathcal{S}(C_1)$ and $R_2\colon V_1\cup V_2 \to \mathcal{S}(C_2)$ where $\mathcal{S}(C_1)$ and $\mathcal{S}(C_2)$ denote the set of all connected subgraphs of some cycles $C_1$ and $C_2$, respectively.
		Moreover, both representations can be chosen so that there are points $p_1 \in V(C_1)$ and $p_2 \in V(C_2)$ such that
		\begin{equation}
			\label{eq1}
			p_1 \in \bigcap_{v\in V_2}R_1(v), \quad p_1 \notin \bigcup_{v \in V}R_1(v), \quad\text{and}\quad p_2\in \bigcap_{v\in V_2}R_2(v),\quad p_2\notin \bigcup_{v\in V_1}R_2(v).
		\end{equation}
		Let $H$ be the graph obtained from $C_1$ and $C_2$ by identifying the points $p_1$ and $p_2$.
		Clearly $H$ is a subdivision of $B$.
		Let $R\colon V\cup V_1\cup V_2 \to \mathcal{S}(H)$ be the mapping defined by
		$$R(v) =
		\begin{cases}
			R_1(v) &\quad \text{if $v \in V$},\\
			R_1(v)\cup R_2(v) &\quad \text{if $v \in V_2$},\\
			R_2(v) &\quad \text{if $v \in V_1$}.
		\end{cases}
		$$
		Since there are no edges between $V$ and $V_1$ in $G'$, it follows from~\eqref{eq1} that $R$ is an $H$-representation of $G'$ and, therefore, $G'$ is an $B$-graph.
		
		It remains to note that both $R_1$ and $R_2$ are proper representations of the corresponding circular-arc graphs.
		Thus, the resulting representation $R$ of $G'$ is also a proper $B$-representation.
	\end{proof}
	
	\begin{lemma}
		\label{lem:isomorphic}
		Let $G$ and $H$ be connected graphs with minimum degree at least three.
		Then, $G$ and $H$ are isomorphic if and only if $G'$ and $H'$ are isomorphic.
	\end{lemma}

	\begin{proof}
		Since subdividing each edge and taking complements preserve the isomorphism relation, it follows that $G \cong H$ if and only if $G_2 \cong H_2$, and, $G_3 \cong H_3$ if and only if $G' \cong H'$.
		Thus, it suffices to prove that $G_2 \cong H_2$ if and only if $G_3 \cong H_3$.
		
		Let $f\colon G \to H$ be an isomorphism and let $f_2 \colon G_2 \to H_2$ be the isomorphism such that $f_2|_{V(G)} = f$, i.e., the restriction of $f_2$ to $V(G)$ equals to $f$.
		We have
		\begin{equation}
			\label{eq2}
			f_2(V(G)) = V(H) \quad \text{and} \quad f_2(V_1(G)) = V_1(H).
		\end{equation}
		
		The graph $G_3$ is constructed by adding all edges between $V(G)$ and $V_1(G)$.
		Likewise, the graph $H_3$ is constructed by adding all edges between $V(H)$ and $V_1(H)$.
		From~\eqref{eq2} it follows that $f_2$ is also an isomorphism from $G_3$ to $H_3$.
		
		For the reverse implication, first note that for any vertex $v$ of $G_3$, we have
		$$
		deg_{G_3}(v) = 
		\begin{cases}
			deg_{G}(v) + |V_1(G)| &\quad \text{if $v\in V(G)$},\\
			2 + |V(G)| &\quad \text{if $v \in V_1(G)$},\\
			2  &\quad \text{if $v \in V_2(G)$}.
		\end{cases}
		$$
		Note that the degree is constant on $V_1(G)$ and on $V_2(G)$.
		An analogous formula is true for every vertex in $H_3$.
		Since the minimum degree in $G$ is three, we have
		$$2 + |V(G)| \geq 6\quad \text{and}\quad deg_G(v) + |V_1(G)|\geq 3 + 3/2|V(G)| > 2 + |V(G)|.$$
		Again, analogous inequalities hold in $H_3$.
		It follows that an isomorphism $f_3 \colon G_3 \to H_3$ satisfies
		$$f_3(V(G)) = V(H), \quad f_3(V_1(G)) = V_1(H) \quad \text{and} \quad f_3(V_2(G)) = V_2(H).$$
		Thus, $f_3$ is also an isomorphism $G_2\to H_2$.
	\end{proof}
	
	\begin{proof}[Proof of Theorem~\ref{thm:8-graphs_gic}]
		$G'$ can be constructed from $G$ in polynomial time.
		Moreover, by Lemma~\ref{lem:8-graph}, $G'$ is a proper $B$-graph.
		Clearly, the isomorphism problem for the class of all graphs with minimum degree three is GI-complete.
		By Lemma~\ref{lem:isomorphic}, $G$ and $H$ with minimum degree three are isomorphic if and only if $G'$ and $H'$ are isomorphic.
		It follows that the isomorphism problem for the class of all proper $B$-graphs is GI-complete.
	\end{proof}

	\section{Conclusions}

	In this paper, we have shown that, for a fixed unicyclic graph $U$ and parameterized by the size of $U$, the recognition problem for proper $U$-graphs and proper Helly $U$-graphs is in \emph{FPT}, and the isomorphism problem for proper $U$-graphs (and therefore, proper Helly $U$-graph) is in \emph{FPT}. We have complemented these positive results with the following hardness results. The recognition problem for $\mathcal{U}$-graphs and proper $\mathcal{U}$-graphs is NP-hard where $\mathcal{U}$ is the class of all unicyclic graphs, and the isomorphism problem for $H$-graphs and proper $H$-graphs is GI-complete when $H$ is not unicyclic. The problems whether general $U$-graphs can be recognized in \emph{FPT}-time and whether they can be tested for isomorphism in \emph{FPT}-time remain open, and we would like to consider these problems as future work.

	\bibliography{Union-bibliography}
	
	\appendix

\end{document}